\documentclass[11pt]{article}

\usepackage{epsf,amsfonts,amsmath, bm}
\usepackage{amsthm}
\usepackage{hyperref}
\usepackage[dvips]{epsfig,graphics}
\usepackage{graphicx}
\usepackage{texdraw}
\usepackage{authblk}
\usepackage{url}
\newtheorem{thm}{Theorem}[subsection]
\newtheorem{theorem}[thm]{Theorem}
\newtheorem{proposition}[thm]{Proposition}
\newtheorem{corollary}[thm]{Corollary}
\newtheorem{remark}[thm]{Remark}
\newtheorem{lemma}[thm]{Lemma}
\DeclareMathOperator{\tr}{tr}

\DeclareMathOperator{\End}{End}

\DeclareMathOperator{\sdet}{sdet}

\title{\textbf{On the $3D$ consistency of a Grassmann extended lattice Boussinesq system}}
\author{Sotiris Konstantinou-Rizos\thanks{skonstantin84@gmail.com, s.konstantinu.rizos@uniyar.ac.ru}}
\affil{Centre of Integrable Systems, P.G. Demidov Yaroslavl State 
University, Russia}

\begin{document}

\maketitle

\begin{abstract}
In this paper, we formulate a ``Grassmann extension'' scheme for constructing noncommutative (Grassmann) extensions of Yang-Baxter maps together with their associated systems P$\Delta$Es, based on the ideas presented in \cite{Sokor-Kouloukas}. Using this scheme, we first construct a Grassmann extension of a Yang-Baxter map which constitutes a lift of a lattice Boussinesq system. The Grassmann-extended Yang-Baxter map can be squeezed down to a novel, integrable, Grassmann lattice Boussinesq system, and  we derive its $3D$-consistent limit. We show that some systems retain their $3D$-consistency property in their Grassmann extension.
\end{abstract}

\hspace{.2cm} \textbf{PACS numbers:} 02.30.Ik 

\hspace{.2cm} \textbf{Mathematics Subject Classification:} 15A75, 35Q53, 39A14, 81R12.

\hspace{.2cm} \textbf{Keywords:} Noncommutative Boussinesq lattice system, Grassmann extensions of Yang-Baxter

\hspace{.2cm} maps, quad-graph systems, Grassmann algebras, Grassmann extensions of discrete integrable sys-

\hspace{.2cm} tems.

\section{Introduction}
Over the past few decades, there has been an increasing interest in the study of noncommutative extensions of integrable equations or systems of equations (indicatively we refer to \cite{Kulish, Dimakis, Grisaru-Penati, Hamanaka-Toda, Lechtenfeld, Sokolov1998}), due to their numerous applications in Physics. Famous examples include noncommutative analogues of the KdV, the NLS, the sine-Gordon and other well-celebrated equations of Mathematical Physics. Therefore, it is quite important to develop methods for solving such -- noncommutative -- systems.

On the other hand, in the commutative case, plenty of methods have been discovered for solving discrete integrable systems (see \cite{Hiet-Frank-Joshi} and the references therein). One of the most well-studied and important class of such systems are the so-called ``quad-graph systems'', namely systems of difference equations defined on an elementary quadrilateral of the two-dimensional lattice. For those quad-graph systems which possess the ``$3D$ consistency'' property, B\"acklund transformations can be derived automatically, and therefore interesting solutions can be constructed starting from trivial ones. Due to the useful properties of $3D$ consistent quad-graph systems and the availability of simple algebraic schemes for contructing solutions to them, they can be used as good models for studying their continuous analogues, i.e. systems of nonlinear PDEs, via continuum limits. At the same time, $3D$ consistent quad-graph systems are strongly related to Yang-Baxter maps,  namely solutions to the set-theoretical Yang-Baxter equation, one of the most fundamental equation of Mathematical Physics. This is a quite important connection, and a lot of work has been done in this direction (indicatively we refer to \cite{ABS-2005, Caudrelier, Pavlos, Pap-Tongas0, Pap-Tongas-Veselov}).

The importance of noncommutative extensions of integrable systems from a Physics perspective, and the innovating results that have already been obtained in the continuously developing field of Discrete Integrable Systems, motivates us to extend to the noncommutative case the already existing methods for constructing solutions to integrable systems in the commutative case. Towards this direction, a few steps have been made over the past few years. In particular, in the recent work of Grahovksi and Mikhailov \cite{Georgi-Sasha}, integrable discretisations were found for a class of NLS equations on Grassmann algebras. This motivated the construction of Grassmann extended systems of differential-difference and difference-difference equations  \cite{Xue-Levi-Liu, Xue-Liu, Xue-Liu-2}, as well as the consideration of continuum limits of Grassmann extended difference equations (see, for instance, \cite{Mao-Liu, Mao-Liu-2}). Furthermore, the latter results and the aforementioned strict relation between quad-graph systems and Yang-Baxter maps motivated the beginning of the extension of the theory of Yang-Baxter maps on Grassmann algebras \cite{GKM, Sokor-Sasha-2}. In addition, a Grassmann extension of the discrete potential KdV equation together with its associate Yang-Baxter map were constructed in \cite{Sokor-Kouloukas}. 

In this paper, motivated by the above-mentioned developments and the results obtained in  \cite{Sokor-Kouloukas}, we formulate a scheme for constructing noncommutative (Grassmann) extensions of quad-graph systems together with their associated Grassmann extended Yang-Baxter maps. Moreover, we answer the main question which arose in \cite{Sokor-Kouloukas} on whether the noncommutativity ``kills'' the $3D$ consistency property for all quad-graph systems. In particular, the Grassmann extended discrete potential KdV system which was constructed in \cite{Sokor-Kouloukas} does not have the $3D$ consistency property. However, this in not the case for all the Grassmann extended integrable systems; in fact, in this paper, we construct a Grassmann extension of a Boussinesq system which retains the $3D$ consistency of its original, commutative version.

As an illustrative example for the description of our scheme, we consider a discrete Boussinesq system. The Boussinesq equation, in both its continuous and its discrete (lattice Boussinesq) version, has been studied extensively over the past few decades, earning its place on the list of fundamental equations of Mathematical Physics. It owes its popularity to its quite interesting and, also, simple form, with a number of applications in Fluid Dynamics and in the theory of Integrable Systems.

\subsection{Main results}
This paper is concerned with the formulation of a scheme for constructing Grassmann extensions of quad-graph systems and their associated Yang-Baxter maps. The methods in this scheme are demonstrated via the following Boussinesq system of difference equations for $p_{n,m}=p(n,m)$, $q_{n,m}=q(n,m)$, $n,m\in\mathbb{N}$:
\begin{align}\label{BSQ-quad}
&(p_{n,m+1}-p_{n+1,m})(p_{n,m}+q_{n,m}q_{n+1,m+1}-r_{n+1,m+1})=(a-b)q_{n,m},\nonumber\\
&(q_{n,m+1}-q_{n+1,m})(p_{n,m}+q_{n,m}q_{n+1,m+1}-r_{n+1,m+1})=b-a,\\
&(r_{n,m+1}-r_{n+1,m})(p_{n,m}+q_{n,m}q_{n+1,m+1}-r_{n+1,m+1})=(b-a)q_{n+1,m+1},\nonumber
\end{align}
where $a,b\in\mathbb{C}$ (see various forms of this system \cite{Bridgman, NPCQ-92, Tongas-Nijhoff}). In fact, we construct and study the integrability of its noncommutive extension not only in terms of possessing a Lax representation, but also as a $3D$ consistent system \cite{Bobenko-Suris, Frank4}. We also derive the associated Yang-Baxter map.

To conclude, we state what is new in this paper:
\begin{enumerate}
\item The formulation of the ideas presented in \cite{Sokor-Kouloukas} into a \textit{Grassmann extension scheme};
\item The derivation of a new Boussinesq-type Yang-Baxter map together with its Grassmann extension;
\item The construction of an integrable, noncommutative (Grassmann) extension of a discrete Boussinesq system and its $3D$-consistent limit. The latter gives rise to the following important point.
\item We show that, for some systems, the $3D$-consistency property does not break in their noncommutative extension.
\end{enumerate}

\subsection{Organisation of the paper}
The paper is organised as follows: The next section provides with preliminary knowledge, essential for the text to be self-contained. In particular, we fix the notation that we use throughout the text, and we give the basic definitions of quad-graph systems and Yang-Baxter maps. Furthermore, we demonstrate the relation between the former and the latter, and the relation between the $3D$ consistency property and the Yang-Baxter equation. We also explain what a Lax representation is for both quad-graph equations and Yang-Baxter maps. Finally, we provide the basic properties of Grassmann algebras, which are essential for this text, and present the basic steps of a simple scheme for constructing Grassmann extensions of discrete integrable systems together with their associated Yang-Baxter maps; the related ideas were discussed in \cite{Sokor-Kouloukas}. In section \ref{Grassmann-BSQ_system}, we apply the aforementioned scheme to system \eqref{BSQ-quad}. Specifically, we consider the associated Yang-Baxter lift of \eqref{BSQ-quad}, for which we construct a Grassmann extension. Then, we show that the latter can be squeezed down to a novel integrable system of lattice equations which can be considered as the Grassmann extension of system \eqref{BSQ-quad}. Finally, in section \ref{3D-consistency}, we present a Boussinesq-type system associated via a conservation law of the one obtained in section \ref{Grassmann-BSQ_system}, and we prove the integrability--in the sense of $3D$-consistency--for a limit of this system. Finally, the last section deals with some concluding remarks and thoughts for future work. 

\section{Preliminaries}\label{Preliminaries}
\subsection{Notation}\label{notation}
Here, we explain the notation we shall be using throughout the text. 

\subsubsection{Functions of discrete variables and shifts}\label{notation}
Let $f$ be a function of two discrete variables $n$ and $m$, i.e $f=f(n,m)$. Let also $\mathcal{S}$ and $\mathcal{T}$ be the shift operators in the $n$ and $m$ direction of a two-dimensional lattice, respectively. We adopt the notation: $f_{00}\equiv f$, $f_{ij}=S^iT^j f$; for example, $f_{10}=f(n+1,m)$, $f_{01}=f(n,m+1)$ and $f_{11}=w(n+1,m+1)$ as represented in Figure \ref{quad-cube}. Furthermore, if our field $f$ lives on the three-dimensional lattice, namely $f=f(n,m,k)$, and $\mathcal{Z}$ is the shift operator in the $k$-direction, then we shall be using three indices to determine the position of $f$ on the lattice. That is, $f_{ijl}=\mathcal{S}^i\mathcal{T}^j \mathcal{Z}^k f(n,m,k)$. For instance, $f_{101}=w(n+1,m,k+1)$ as in Figure \ref{quad-cube}.

\begin{figure}[ht]
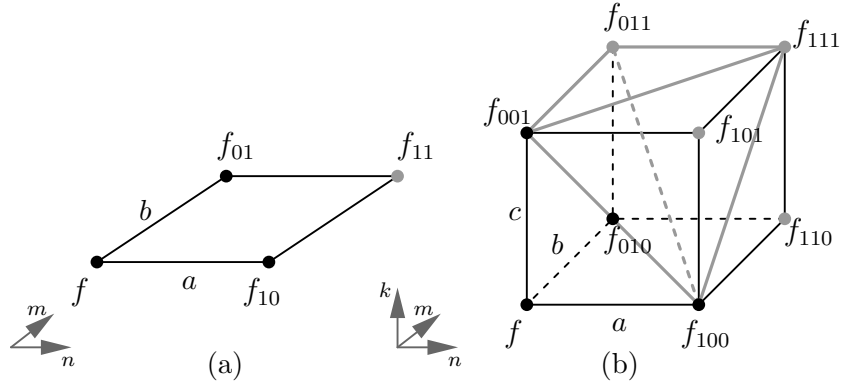

\centering
\centertexdraw{ 
\setunitscale 0.45
\move (-3 0.5)  \lvec(-1 0.5) \lvec(.5 1.5) \lvec(-1.5 1.5) \lvec(-3 0.5)
\lpatt()
\move (2 0)  \lvec(4 0) \lvec(5 1) 
\lpatt(0.067 0.1) \lvec(3 1) \lvec(2 0) 
\lpatt() \lvec(2 2) \lvec(3 3) 
\lpatt (0.067 0.1) \lvec(3 1) 
\lpatt() \move (3 3) \lvec(5 3) \lvec(4 2) \lvec(4 0) 
\move (2 2) \lvec(4 2)
\move (5 3) \lvec(5 1)

\htext (2.9 .6) {\large$f_{010}$}
\move(2 2)\setgray 0.6 \linewd 0.04 \lvec(4 0)  \lpatt (0.067 0.1) \lvec(3 3)  \lpatt ()
\move(4 0)\setgray 0.6 \linewd 0.04 \lvec(5 3)\lvec(2 2)
\move(5 3)\setgray 0.6 \linewd 0.04 \lvec(3 3)\lvec(2 2)

\move (-3 .5) \fcir f:0.0 r:0.075
\move (-1 .5) \fcir f:0.0 r:0.075
\move (.5 1.5) \fcir f:0.6 r:0.075
\move (-1.5 1.5) \fcir f:0.0 r:0.075
\move (2 0) \fcir f:0.0 r:0.075
\move (4 0) \fcir f:0.0 r:0.075
\move (5 1) \fcir f:0.6 r:0.075

\move (2 2) \fcir f:0.0 r:0.075
\move (4 2) \fcir f:0.6 r:0.075
\move (5 3) \fcir f:0.6 r:0.075
\move (3 3) \fcir f:0.6 r:0.075
\move (3 1) \fcir f:0.0 r:0.075

\move(-4 -.5) \linewd 0.02 \setgray 0.4 \arrowheadtype t:F \avec(-3.3 -.5) 
\move(-4 -.5) \linewd 0.02 \setgray 0.4 \arrowheadtype t:F \avec(-3.5 -.1) 
\htext (-3.4 -.7) {\scriptsize{$n$}}
\htext (-3.8 -.1) {\scriptsize{$m$}}

\move(0.5 -.5) \linewd 0.02 \setgray 0.4 \arrowheadtype t:F \avec(1.2 -.5) 
\move(0.5 -.5) \linewd 0.02 \setgray 0.4 \arrowheadtype t:F \avec(1 -.1) 
\move(0.5 -.5) \linewd 0.02 \setgray 0.4 \arrowheadtype t:F \avec(.5 .2) 
\htext (1.1 -.7) {\scriptsize{$n$}}
\htext (.7 -.1) {\scriptsize{$m$}}
\htext (.3 .1) {\scriptsize{$k$}}

\htext (-3.3 0) {\large$f$}
\htext (-1.3 0) {\large$f_{10}$}
\htext (-1.6 1.7) {\large$f_{01}$}
\htext (.5 1.7) {\large$f_{11}$}
\htext (-2 0.2) {$a$}
\htext (-2.5 1) {$b$}

\htext (1.75 -.5) {\large$f$}
\htext (3.8 -.5) {\large$f_{100}$}
\htext (1.5 2.1) {\large$f_{001}$}
\htext (5 .7) {\large$f_{110}$}
\htext (4.2 1.9) {\large$f_{101}$}
\htext (5.1 3) {\large$f_{111}$}
\htext (2.9 3.2) {\large$f_{011}$}

\htext (3 -.3) {$a$}
\htext (2.3 0.6) {$b$}
\htext (1.8 1) {$c$}

\textref h:C v:C \htext(3.1 -0.7){(b)}
\textref h:C v:C \htext(-1.5 -0.7){(a)}
}
\caption{(a) Elementary square of the $2D$ lattice and (b) elementary cube of the $3D$ lattice.}\label{quad-cube}
\end{figure}

\subsubsection{Commutative and anticommutative variables} 
We shall be using Latin letters for all commuting variables, whereas all the anticommutative variables will be denoted by Greek letters. For instance, $pq=qp$, whereas $\tau\theta=-\theta\tau$. As an exception, the spectral parameter, $\lambda\in\mathbb{C}$, is a commuting variable.

\subsection{$3D$ consistency VS the Yang-Baxter equation}
``Quad-graph'' equations (or systems) and ``Yang-Baxter maps'' constitute the two sides of the same coin. In this section, we explain the relation between the $3D$ consistency property and the Yang-Baxter equation.

\subsubsection{Quad-graph equations and parametric Yang-Baxter maps}
Using the notation introduced in section \ref{notation}, let the fields ($f$, $f_{10}$, $f_{01}$, $f_{11}$) lie on the vertices of the square in Figure \ref{quad-cube}. Let us also consider the following equation
\begin{equation}\label{eq-quad-graph}
Q(f,f_{10},f_{01},f_{11};a,b)=0,
\end{equation}
where the parameters $a,b\in\mathbb{C}$ and $Q$ is a linear function in every field $f_{ij}$. Equation \eqref{eq-quad-graph} is called \textit{equation on quad-graph} and can be interpreted as in Figure \ref{quad-cube}-(a). That is, knowing any 3 of the fields $f_{ij}$ on the vertices, one can uniquely identify the fourth, using \eqref{eq-quad-graph}.

Now, by the term ``parametric Yang-Baxter map'' we understand set-theoretical solutions of the parametric Yang-Baxter equation, namely maps $Y_{a,b}\in\End(V\times V)$, where $V$ is algebraic variety, i.e.
\begin{equation}\label{Ymap}
(x,y)\stackrel{Y_{a,b}}{\mapsto}\left(u(x,y;a,b),v(x,y;a,b)\right),
\end{equation}
satisfying the parametric Yang-Baxter equation
\begin{equation}\label{YB_eq1}
Y^{23}_{b,c}\circ Y^{13}_{a,c} \circ Y^{12}_{a,b}=Y^{12}_{a,b}\circ Y^{13}_{a,c} \circ Y^{23}_{b,c}.
\end{equation}
The $Y^{ij}\in\End(V\times V\times V)$ are defined as: $Y^{12}_{a,b}=Y_{a,b}\times id$, $Y^{23}_{b,c}=id\times Y_{a,b}$ and $Y^{13}_{a,c}=\pi^{12} Y^{23}_{b,c} \pi^{12}$, where $\pi^{12}$ is the involution defined by $\pi^{12}((x;a),(y;b),(z;c))=((y;b),(x;a),(z;c))$. The geometric interpretation of these maps, can be understood in a similar way as for quad-graph equations, but with the values being considered on the edges of the quad, as in Figure \ref{quad-cube}-(b).

Similarly to quad-graph systems, we can also interpret Yang-Baxter maps on the square, but considering the values on the edges instead of the vertices (see Figure \ref{quadToYB}-(b)).

\begin{figure}[ht]
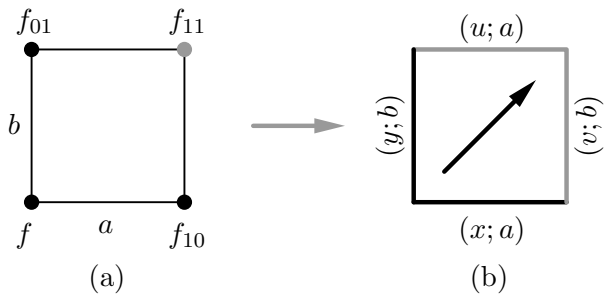

\centering
\centertexdraw{ 
\setunitscale .8
\move (0 0)  \lvec (1 0) \lvec(1 1) \lvec(0 1) \lvec(0 0)
\linewd 0.03 \move (1.45 .5)  \setgray 0.6 \arrowheadtype t:F \avec(2.05 0.5) \linewd 0 \setgray 0
\move (2.5 0) \linewd 0.03 \lvec (3.5 0)  \setgray 0.6 \lvec(3.5 1)  \lvec(2.5 1)  \setgray 0 \lvec(2.5 0) \linewd 0
\move (0 0) \fcir f:0.0 r:0.05
\move (1 0) \fcir f:0.0 r:0.05
\move (0 1) \fcir f:0.0 r:0.05
\move (1 1) \fcir f:0.6 r:0.05
\linewd 0.03 \move (2.7 0.2) \arrowheadtype t:F \avec(3.3 0.8) \linewd 0
\htext (-.1 -.3) {$f$}
\htext (.9 -.3) {$f_{10}$}
\htext (-.1 1.1) {$f_{01}$}
\htext (.9 1.1) {$f_{11}$}
\htext (0.44 -.2) {$a$}
\htext (-.15 .45) {$b$}
\htext (2.8 -.28) {$(x;a)$}
\vtext (2.4 .3) {$(y;b)$}
\htext (2.8 1.04) {$(u;a)$}
\vtext (3.7 .3) {$(v;b)$}

\textref h:C v:C \htext(3 -0.5){(b)}
\textref h:C v:C \htext(.5 -0.5){(a)}
}
\caption{Initial values on the (a) vertices, (b) edges.}\label{quadToYB}
\end{figure}

\subsubsection{Lax representations \& integrability}
Equation \eqref{eq-quad-graph} admits \textit{quad-Lax representation}, if there is a (Lax) matrix $L_a=L_a(f,f_{10};\lambda)\equiv L_a(f,f_{10})$\footnote{We usually skip writing explicitly the dependence on the spectral parameter $\lambda$.}, $\lambda\in\mathbb{C}$, such that
\begin{equation*}
L_a\left(\mathcal{T}f,\mathcal{T}f_{10}\right)L_b\left(f,f_{01}\right)=L_b\left(\mathcal{S}f,\mathcal{S}f_{01}\right)L_a\left(f,f_{10}\right),
\end{equation*}
where $\mathcal{S}$ and $\mathcal{T}$ are shift operators, as defined in section \ref{notation}. 

Similarly, for Yang-Baxter maps, Lax matrix is a matrix $L=L(x,a;\lambda)\equiv L_a(x)$ that satisfies the following matrix refactorisation problem \cite{Veselov2}
\begin{equation}\label{eqLax}
L_a(u)L_b(v)=L_b(y)L_a(x).
\end{equation}
If equation \eqref{eqLax} defines a map \eqref{Ymap}, then it is called Lax representation of the map. An alternative way to verify that a map satisfies the Yang-Baxter equation is to consider the following matrix trifactorisation problem
\begin{equation*}
L_a(u)L_b(v)L_c(w)=L_a(x)L_b(y)L_c(z),
\end{equation*}
where $L_a(x)$ is the same matrix satisfying \eqref{eqLax}. In particular, if the above  trifactorisation problem implies that $u=x$, $v=y$ and $w=z$, then map \eqref{Ymap} defined by \eqref{eqLax} satisfies the parametric Yang-Baxter equation \eqref{YB_eq1} \cite{Kouloukas2, Veselov}.

In the case of quad-graph equations or systems as \eqref{eq-quad-graph}, the possession of Lax representation is usually used as working definition of integrability. However, a stronger integrability criterion is that of $3D$-consistency \cite{Bobenko-Suris, Frank4} which implies integrability in the sense of Lax representation. 

From the analysis-point-of-view, $3D$-consistency is the property of equation \eqref{eq-quad-graph} to be consistently generalisable in three dimensions, by ``adding'' a third discrete variable $k$ in the field $f$, namely considering $f=f(n,m,k)$. Geometrically, it means that a quad-graph system--as interpreted in Figure \ref{quad-cube}-(a)--can be generalised and ``written'' in a consistent way on the faces of the cube of Figure \ref{quad-cube}-(b). That is, we first rewrite our system \eqref{eq-quad-graph} on the bottom, front and left side of the cube, respectively, as follows:
\begin{equation}\label{eqsOnCube}
Q(f,f_{100},f_{010},f_{110};a,b)=0,~~~Q(f,f_{100},f_{001},f_{101};a,c)=0,~~~Q(f,f_{001},f_{010},f_{011};c,b)=0.
\end{equation}
Then, considering $f$, $f_{100}$, $f_{010}$ and $f_{001}$ as initial values on the cube in Figure \ref{quad-cube}-(b), there are three ways to calculate $f_{111}$: 1) Using the first equation of \eqref{eqsOnCube}, determine $f_{110}$; 2) Using the second equation, determine $f_{101}$; 3) With use of the last equation of \eqref{eq-quad-graph}, determine $f_{011}$. Consequently, having $f_{110}$, $f_{101}$ and $f_{011}$ at our disposal, we can determine $f_{111}$, using any of the top, back or right side of the cube. $3D$-consistency means that, independently of which of the former sides we use, we obtain exactly the same value $f_{111}$.

The strict relation between the $3D$-consistency property and the Yang-Baxter equation can be demonstrated in Figure \ref{YB-3D}. In fact, one can consider three initial values $x$, $y$ and $z$ taken on the sides of the cube as in Figure \ref{YB-3D}. Now, acting on $(x,y,z)$ with the left part of the Yang-Baxter equation, that is, using the bottom, back and left side of the cube, we obtain new values $(\hat{\hat{x}},\hat{\hat{y}},\hat{\hat{z}})$. On the other hand, acting on $(x,y,z)$ with the right part of the Yang-Baxter equation, namely via the left, front and top side of the cube, we obtain the values $(\tilde{\tilde{x}},\tilde{\tilde{y}},\tilde{\tilde{z}})$. The Yang-Baxter equation is satisfied when the ``hats'' coincide with the ``tildes'' and vice versa.

\begin{figure}[ht]
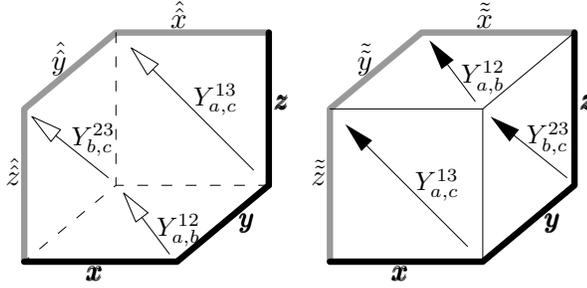

\centering
\centertexdraw{ 
\setunitscale .8
\move(1.6 1.5) \setgray 0.6 \linewd 0.04 \lvec (0.6 1.5) \lvec(0 1) \lvec(0 0) \setgray 0 \linewd 0
\move (0 0) \linewd 0.04 \lvec (1 0) \lvec(1.6 .5)\linewd 0 \lpatt(0.067 0.09) \lvec(0.6 .5) \lvec(0 0) 
\lpatt() \move(1.6 .5)\linewd 0.04 \lvec(1.6 1.5) \linewd 0
\lpatt(0.067 0.09)  \move(0.6 1.5)\lvec(0.6 .5) \lpatt()
\move(2 0)\setgray 0.6 \linewd 0.04 \lvec(2 1)\lvec(2.6 1.5)  \lvec(3.6 1.5)\setgray 0 \linewd 0
\move (2 0) \linewd 0.04 \lvec (3 0) \lvec(3.6 .5) \lvec(3.6 1.5) \linewd 0 \lvec(3 1) \lvec(3 0) 
\move(2 1)\lvec(3 1)
\move (1.5 .6)  \arrowheadtype t:T \avec(0.7 1.4) \linewd 0 \setgray 0 
\move (0.95 .05)  \arrowheadtype t:T \avec(0.65 0.45) \linewd 0 \setgray 0
\move (0.55 .55)  \arrowheadtype t:T \avec(0.05 0.95) \linewd 0 \setgray 0

 \move (2.9 .1)  \arrowheadtype t:F \avec(2.1 0.9) \linewd 0 \setgray 0
 \move (3.55 .55)  \arrowheadtype t:F \avec(3.05 0.95) \linewd 0 \setgray 0
\move (2.95 1.05)  \arrowheadtype t:F \avec(2.65 1.45) \linewd 0 \setgray 0
\htext (2.4 -.12) {\pmb{$x$}}
\htext (3.4 0.19) {\pmb{$y$}}
\htext (3.63 1) {\pmb{$z$}}
\htext (2.95 1.53) {$\tilde{\tilde{x}}$}
\htext (2.17 1.22) {$\tilde{\tilde{y}}$}
\htext (1.87 .5) {$\tilde{\tilde{z}}$}

\htext (.4 -.12) {\pmb{$x$}}
\htext (1.4 0.19) {\pmb{$y$}}
\htext (1.63 1) {\pmb{$z$}}
\htext (0.95 1.53) {$\hat{\hat{x}}$}
\htext (0.17 1.22) {$\hat{\hat{y}}$}
\htext (-.13 .5) {$\hat{\hat{z}}$}

\htext (2.55 .4) {\small$Y^{13}_{a,c}$}
\htext (3.28 .7) {\small$Y^{23}_{b,c}$}
\htext (2.85 1.1) {\small$Y^{12}_{a,b}$}

\htext (1.1 .96) {\small$Y^{13}_{a,c}$}
\htext (0.86 .1) {\small$Y^{12}_{a,b}$}
\htext (0.3 .69) {\small$Y^{23}_{b,c}$}

}
\caption{Yang-Baxter equation. Geometric interpretation.}\label{YB-3D}
\end{figure}

\subsection{Grassmann algebra}
Consider $G$ to be a $\mathbb{Z}_2$-graded algebra over $\mathbb{C}$. Thus, $G$, as a linear space, is a direct sum $G=G_0\oplus G_1$ (mod 2), such that $G_iG_j\subseteq G_{i+j}$. The elements of $G$ that belong either to $G_0$ or to $G_1$ are called \textit{homogeneous}, the ones in $G_1$ are called \textit{odd} (or {\sl fermionic}), while those in $G_0$ are called \textit{even} (or {\sl bosonic}).

By definition, the parity $|a|$ of an even homogeneous element $a$ is $0$, while it is $1$ for odd homogeneous elements. The parity of the product $|ab|$ of two homogeneous elements is a sum of their parities: $|ab|=|a|+|b|$. Now, for any homogeneous elements $a$ and $b$, Grassmann commutativity means that $ba=(-1)^{|a||b|}ab$ . This implies that if $\alpha\in G_1$, then $\alpha^2=0$, and 
$\alpha a=a \alpha$, for any $a\in G_0$. 

The notions of the determinant and the trace of a matrix in $G$ are defined for square matrices, $M$, of the following block-form
\begin{equation}\label{block-form}
M=\left(
\begin{matrix}
 P & \Pi \\
 \Lambda & L
\end{matrix}\right).
\end{equation}
The blocks $P$ and $L$ are matrices with even entries, while $\Pi$ and $\Lambda$ possess only odd entries (note that the block matrices are not necessarily square matrices).  In particular, the \textit{superdeterminant} of $M$, which is usually denoted by $\sdet(M)$, is defined to be the following quantity
\begin{equation*}
\sdet(M)=\det(P-\Pi L^{-1}\Lambda)\det(L^{-1})=\det(P^{-1})\det(L-\Lambda P^{-1}\Pi),
\end{equation*}
where $\det(\cdot)$ is the usual determinant of a matrix.

In this section, we gave all the definitions related to Grassmann algebras that are essential for this paper. However, for more information on Grassmann analysis one can consult \cite{Berezin}.

\subsection{Grassmann extension scheme}\label{Grassmann extension scheme}
Here, we demonstrate a scheme for constructing Grassmann extensions of discrete integrable systems together with their associated Grassmann extended Yang-Baxter maps. We formulate the ideas presented in \cite{Sokor-Kouloukas} which constitute a combination of the methods introduced in \cite{Pap-Tongas} and  \cite{GKM}.

\begin{figure}[ht]
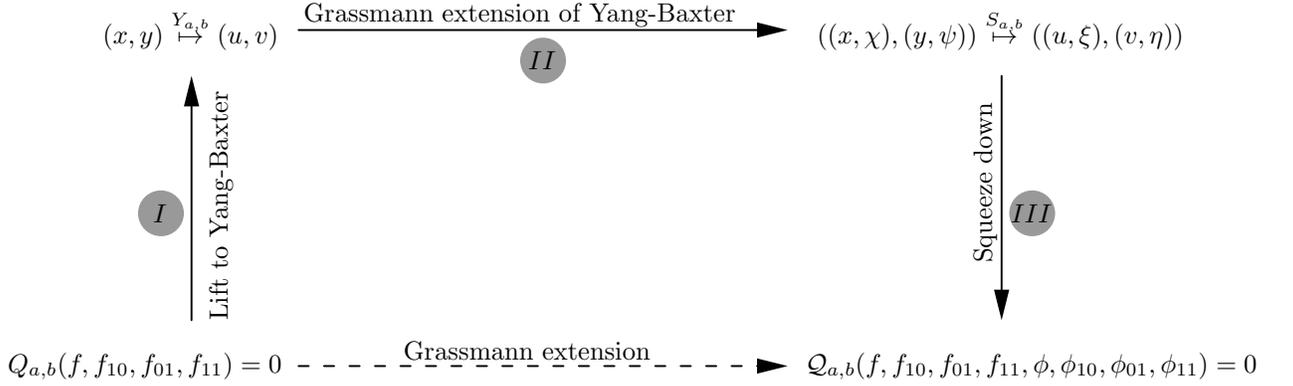

\centering
\centertexdraw{ 
\setunitscale 0.8
\move (-1.6 2.2)  \arrowheadtype t:F \avec(1.6 2.2)
\move (-2.3 0.3)  \arrowheadtype t:F \avec(-2.3 1.9)
\move (3 1.9)  \arrowheadtype t:F \avec(3 0.3) 
\textref h:C v:C \htext(-2.6 0){\small $Q_{a,b}(f,f_{10},f_{01},f_{11})=0$}
\textref h:C v:C \htext(3.2 0){\small $\mathcal{Q}_{a,b}(f,f_{10},f_{01},f_{11},\phi,\phi_{10},\phi_{01},\phi_{11})=0$}
\textref h:C v:C \htext(-2.3 2.2){\small $(x,y)\stackrel{Y_{a,b}}{\mapsto}(u,v)$}
\textref h:C v:C \htext(3 2.2){\small $((x,\chi),(y,\psi))\stackrel{S_{a,b}}{\mapsto}((u,\xi),(v,\eta))$}\lpatt(0.067 0.1)
\move (-1.6 0)  \arrowheadtype t:F \avec(1.6 0)
\move (-2.5 1)\fcir f:.6 r:.15
\textref h:C v:C \small{\htext(-2.5 1){$I$}}
\move (3.2 1)\fcir f:.6 r:.15
\textref h:C v:C \small{\htext(3.2 1){$III$}}
\move (0 2)\fcir f:.6 r:.15
\textref h:C v:C \small{\htext(0 2){$II$}}
\textref h:C v:C \small{\htext(-.1 .1){\small Grassmann extension}}
\textref h:C v:C \small{\htext(-.15 2.3){\small Grassmann extension of Yang-Baxter}}
\textref h:C v:C \small{\vtext(-2.1 1.05){\small Lift to Yang-Baxter}}
\textref h:C v:C \small{\vtext(2.9 1.2){\small Squeeze down}}
}
\caption{Grassmann extension scheme}\label{Grassmannscheme}
\label{sqr-cube}
\end{figure}

The scheme consists of three steps:\\
\textbf{I.} Starting from an integrable quad-graph equation, $Q_{a,b}(f,f_{10},f_{01},f_{11})=0$, derive a Yang-Baxter map using the symmetries of the equation in order to transfer from a one-field equation (with field $f$) to a two-field map (with fields $u$ and $v$) \cite{Pap-Tongas}, namely a map $(x,y)\stackrel{Y_{a,b}}{\mapsto}(u,v)$. Note that this method is reversible, when applied to equations of certain form.\\
\textbf{II.} The method was introduced in \cite{GKM, Georgi-Sasha}. Starting from map $Y_{a,b}$ obtained in \textbf{Step I}, construct its noncommutative--Grassmann--extension, namely map $((x,\chi),(y,\psi))\stackrel{S_{a,b}}{\mapsto}((u,\xi),(v,\eta))$. This extension is applicable to Yang-Baxter maps which admit Lax matrix, and it is based on the consideration of a more general Lax matrix which includes anticommutative variables. That is, we consider an augmented Lax matrix, $\mathcal{L}_a=\mathcal{L}_a(x,\chi)$, which contains the old (bosonic) ``$x$'' and the new (fermionic) ``$\chi$'' variables. Our demand is that this matrix satisfies two conditions: 1. Its bosonic limit is equal to the original Lax matrix (with only bosonic elements); that is, $\lim_{\chi\to 0}\mathcal{L}_a(x,\chi)=L_a(x)$.  2. Its superdeterminant is equal to the determinant of the original Lax matrix, i.e. $\sdet{\mathcal{L}_a}=\det{L_a}$.  Note that the augmented matrix $\mathcal{L}_a$ must be in the block-form \eqref{block-form} in order to be able to define its determinant. This method will be demonstrated in the next section for our map.\\
\textbf{III.} Since \textbf{Step I} is reversible, we can apply the reverse idea to the map $S_{a,b}$ obtained in \textbf{Step II} in order to ``squeeze it down'' to a lattice equation $\mathcal{Q}_{a,b}(f,f_{10},f_{01},f_{11},\phi,\phi_{10},\phi_{01},\phi_{11})=0$, such that $\lim_{(\phi,\phi_{10},\phi_{01},\phi_{11})\to 0} \mathcal{Q}_{a,b} = Q_{a,b}$. To do so, we use some symmetries of map $S_{a,b}$. The derived Grassmann extended quad-graph system, $ \mathcal{Q}_{a,b} =0$, is by definition integrable, since it has a Lax representation. Its Lax representation can be derived from the matrix refactorisation problem associated with the Grassmann extended Yang-Baxter map $ \mathcal{S}_{a,b} $, by relabeling the variables (matrix entries).

\section{Boussinesq system and a lift to a Boussinesq type Yang-Baxter map}\label{Grassmann-BSQ_system}
In this section, starting from a Boussinesq lattice system, we construct its associated Yang-Baxter lift.

\subsection{Boussinesq lattice equation}
The lattice Boussinesq system \eqref{BSQ-quad}, in the notation introduced in section \ref{notation}, reads
\begin{align}
(p_{01}-p_{10}) (p-r_{11}+q q_{11})&=(a-b)q,\nonumber\\
(q_{01}-q_{10}) (p-r_{11}+q q_{11})&=b-a,\label{BSQ-system}\\
(r_{01}-r_{10}) (p-r_{11}+q q_{11})&=(b-a)q_{11},\nonumber
\end{align}
where $a,b\in\mathbb{C}$, and it possesses the following strong Lax representation 
\begin{equation}\label{LaxEq-BSQ}
L_a(p_{01},q_{01},q_{11},r_{11})L_b(p,q,q_{01},r_{01})=L_b(p_{10},q_{10},q_{11},r_{11})L_a(p,q,q_{10},r_{10}),
\end{equation}
where $L_a$ is given by the following $3\times 3$ matrix \cite{Tongas-Nijhoff}
\begin{equation}\label{Lax-BSQ}
L_a(p,q,q_{10},r_{10}):=
\left(
\begin{matrix}
 -q_{10} & 1 & 0 \\
 -r_{10} & 0 & 1 \\
a-p q_{10}-q r_{10}-\lambda & p & q
\end{matrix}\right).
\end{equation}

\subsection{Step \textit{I}: Lift to a Yang-Baxter map}
Our aim is to derive a Yang-Baxter map starting from \eqref{BSQ-system}. The idea is to move from the fields $(p,q,r)$ (functions of two discrete variables $n,m\in\mathbb{N}$) to elements of an algebraic variety $V$. The right change of variables is indicated by the Lax representation \eqref{Lax-BSQ} itself.  

In particular, comparing \eqref{LaxEq-BSQ} to the following matrix refactorisation problem
\begin{equation}\label{LaxEq-BSQ-YB}
L_a(u_1,u_2,u_3,u_4)L_b(v_1,v_2,v_3,v_4)=L_b(y_1,y_2,y_3,y_4)L_a(x_1,x_2,x_3,x_4),
\end{equation}
we set $x_1=p$, $x_2=q$, $x_3=\mathcal{S}q$ and $x_4=\mathcal{S}r$, namely we consider the following $3\times 3$ matrix 
\begin{equation}\label{Lax-BSQ-YB}
L_a(\pmb{x}):=
\left(
\begin{matrix}
 -x_3 & 1 & 0 \\
 -x_4 & 0 & 1 \\
a-x_1x_3-x_2x_4-\lambda & x_1 & x_2
\end{matrix}\right), \quad \pmb{x}:=(x_1,x_2,x_3,x_4).
\end{equation}

Here, we understand $x_i$, $i=1,\ldots,4$, as elements of an algebraic variety $V$, and we substitute \eqref{Lax-BSQ-YB} to \eqref{LaxEq-BSQ-YB}. Then, \eqref{LaxEq-BSQ-YB} implies a correspondence given by

\begin{minipage}{.5\linewidth}
\begin{align*}
  u_1&=y_1+\frac{a-b}{x_1-y_4+x_2y_3}x_2,\\
	u_2&=y_2+\frac{b-a}{x_1-y_4+x_2y_3},\\
	u_3&=y_3,\\
	u_4&=y_4+v_1-x_1,
\end{align*}
\end{minipage}
\hspace{-.9 cm}
\begin{minipage}{.5\linewidth}
\begin{align}
  v_2&=x_2,\nonumber\\
	v_3&=x_3+\frac{b-a}{x_1-y_4+x_2y_3},\label{correspondence-BSQ}\\
	v_4&=x_4+\frac{b-a}{x_1-y_4+x_2y_3}y_3.\nonumber
\end{align}
\end{minipage}\\
This correspondence is a solution of \eqref{LaxEq-BSQ-YB} for any $v_1$. For a particular value of $v_1$, the above correspondence defines the following eight-dimensional map; in fact, we have the following.

\begin{proposition}
The map 
\begin{equation}\label{YB-BSQ}
(\textbf{x},\textbf{y})\stackrel{Y_{a,b}}{\rightarrow}(\textbf{u},\textbf{v}), 
\end{equation}
given by
 
\begin{minipage}{.5\linewidth}
\begin{align*}
 x_1\mapsto u_1&=y_1+\frac{a-b}{x_1-y_4+x_2y_3}x_2,\\
 x_2\mapsto u_2&=y_2+\frac{b-a}{x_1-y_4+x_2y_3},\\
 x_3\mapsto u_3&=y_3,\\
 x_4\mapsto u_4&=y_4,
\end{align*}
\end{minipage}
\hspace{-.9 cm}
\begin{minipage}{.5\linewidth}
\begin{align}\label{YB-Boussinesq}
y_1\mapsto v_1&=x_1, \nonumber\\
y_2\mapsto v_2&=x_2,\nonumber\\
y_3\mapsto v_3&=x_3+\frac{b-a}{x_1-y_4+x_2y_3},\\
y_4\mapsto v_4&=x_4+\frac{b-a}{x_1-y_4+x_2y_3}y_3.\nonumber
\end{align}
\end{minipage}
\\
is an eight-dimensional parametric Yang-Baxter map with invariants
\begin{subequations}
\begin{align}
I_1&=x_2+y_2-x_3-y_3,\\
I_2&=x_1+y_1+x_2y_2,\\
I_3&=x_4+y_4-x_3y_3,\\
I_4&=b(x_2-x_3)-a(y_3-y_2)+(x_4-x_3y_2-y_1)(x_1+x_2y_3-y_4).
\end{align}
\end{subequations}
\end{proposition}
\begin{proof}
For the particular choice $v_1=x_1$, the correspondence \eqref{correspondence-BSQ} implies the eight-dimensional map \eqref{YB-BSQ}-\eqref{YB-Boussinesq}. The Yang-Baxter property of the latter can be shown by straightforward substitution of \eqref{YB-Boussinesq} into the Yang-Baxter equation.

Regarding the invariants, the trace of the monodromy matrix $\tr(L_b(\textbf{y})L_a(\textbf{x}))=1+I_2-I_3$, where $L_a$ is given in \eqref{Lax-BSQ-YB}. Thus, $I_2$ and $I_3$ are invariants. The rest, $I_1$ and $I_4$, are found from the characteristic equation $\det(L_b(\textbf{y})L_a(\textbf{x})-k\cdot  \mathbb{I}_3)$, where $\mathbb{I}_3$ is the $3\times 3$ identity matrix.
\end{proof}

\begin{remark}\normalfont
The above procedure is reversible. That is, starting from map \eqref{YB-BSQ}-\eqref{YB-Boussinesq}, we can retrieve the Boussinesq lattice system \eqref{YB-BSQ}-\eqref{YB-Boussinesq}. This follows from the observation that $x_3=y_2$, in \eqref{YB-Boussinesq}, implies $u_2=v_3$, in combination with a certain change of variables.
\end{remark}

\subsection{Step \textit{II}: Grasmann extended Yang-Baxter map of Boussinesq type}
In this section, we construct a Grassmann extension of Boussinesq type Yang-Baxter map \eqref{YB-BSQ}-\eqref{YB-Boussinesq}. In order to do that, we consider the Lax matrix \eqref{Lax-BSQ} augmented with two additional fermionic fields $\chi_1, \chi_2$, such that the conditions described in step II of the scheme are satisfied. 

In particular, we consider following $4\times 4$ matrix
\begin{equation}\label{Lax-BSQ-YB-G}
\mathcal{L}_a(\pmb{x},\pmb{\chi}):=
\left(
\begin{matrix}
 -x_3 & 1 & 0 & 0 \\
 -x_4 & 0 & 1 & 0 \\
a-x_1x_3-x_2x_4-\chi_1\chi_2-\lambda & x_1 & x_2 & \chi_1\\
-\chi_2 & 0 & 0 & 1
\end{matrix}\right), \quad (\pmb{x},\pmb{\chi}):=(x_1,x_2,x_3,x_4,\chi_1,\chi_2),
\end{equation}
which is matrix $L_a$ in \eqref{Lax-BSQ-YB} augmented with two additional fields $\chi_i\in G_1$, $i=1,2$. The above generalisation respects the following conditions 
\begin{enumerate}
	\item Bosonic limit
	\begin{equation}\label{cond2}
	\lim_{\pmb{\chi}\to 0} \mathcal{L}_a(\pmb{x},\pmb{\chi}) =L_a(\pmb{x});
	\end{equation}
	\item Determinant
	\begin{equation}\label{cond1}
	\sdet(\mathcal{L}_a)\footnote{By $\sdet(.)$ we denote the ``superdeterminant'' in the Grassmann case \cite{Berezin}. The superdeterminant is defined on matrices of the block form $M=\left(
\begin{matrix}
 P & \Pi \\
 \Lambda & L
\end{matrix}\right)$,
where $P$ and $L$ are square matrices of even entries, whereas $\Pi$ and $\Lambda$ are matrices with odd entries, not necessarily square.}=\det(L_a)=a-\lambda.
	\end{equation}
\end{enumerate}

\begin{proposition}
The matrix refactorisation problem
\begin{equation}\label{eqLax-BSQ-YB-G}
\mathcal{L}_a(\pmb{u},\pmb{\xi})\mathcal{L}_b(\pmb{v},\pmb{\eta})=\mathcal{L}_b(\pmb{y},\pmb{\psi})\mathcal{L}_a(\pmb{x},\pmb{\chi}),
\end{equation}
where $\mathcal{L}_a=\mathcal{L}_a(\textbf{x},\pmb{\chi})$ is given by \eqref{Lax-BSQ-YB-G}, is equivalent to the following correspondence:\\
\begin{minipage}{.5\linewidth}
\begin{subequations}\label{cor-G}
\begin{align}
 u_1&=y_1+\frac{a-b}{x_1-y_4+x_2y_3+\chi_1\psi_2}x_2,\label{cor-G-a}\\
 u_2&=y_2+\frac{b-a}{x_1-y_4+x_2y_3+\chi_1\psi_2},\label{cor-G-b}\\
 u_3&=y_3,\label{cor-G-c}\\
 \xi_1&=\psi_1+\frac{a-b}{x_1-y_4+x_2y_3}\chi_1,\label{cor-G-d}\\
 \xi_2&=\psi_2, \label{cor-G-e}
\end{align}
\end{subequations}
\end{minipage}
\begin{minipage}{.5\linewidth}
\begin{subequations}\label{cor2-G}
\begin{align}
 v_2&=x_2,\label{cor2-G-a}\\
 v_3&=x_3+\frac{b-a}{x_1-y_4+x_2y_3+\chi_1\psi_2},\label{cor2-G-b}\\
 v_4&=x_4+\frac{b-a}{x_1-y_4+x_2y_3+\chi_1\psi_2}y_3,\label{cor2-G-c}\\
 \eta_1&=\chi_1,\label{cor2-G-d}\\
 \eta_2&=\chi_2+\frac{a-b}{x_1-y_4+x_2y_3}\psi_2,\label{cor2-G-e}
\end{align}
\end{subequations}
\end{minipage}
\\
and
\begin{equation}\label{eq-u1-v4}
u_4=y_4+v_1-x_1.
\end{equation}
\end{proposition}
\begin{proof}
Equation \eqref{eqLax-BSQ-YB-G} implies 
\begin{equation*}
u_3=y_3, \quad v_2=x_2, \quad \xi_2=\psi_2, \quad \chi_1=\chi_1,
\end{equation*}
equation \eqref{eq-u1-v4} for $u_4$ and $v_1$, as well as the following system of equations
\begin{subequations}\label{cor-sys-uv}
\begin{align}
&v_3-u_2=x_3-y_2,\label{cor-sys-uv-a}\\
&v_3y_3-v_4=y_2x_3-x_4,\label{cor-sys-uv-b}\\
&u_1+u_2x_2=y_1+y_2x_2,\label{cor-sys-uv-c}\\
&\xi_1+u_2\eta_1=y_2\chi_1+\psi_1,\label{cor-sys-uv-d}\\
&\eta_2-\psi_2 v_3=\chi_2-\psi_2x_3,\label{cor-sys-uv-e}\\
&v_3(u_4-v_1)+b-x_2v_4-\eta_1\eta_2=x_3(y_4-x_1)+a-x_2x_4-\chi_1\chi_2,\label{cor-sys-uv-f}\\
&u_2(v_1-u_4)+a-u_1y_3-\xi_1\psi_2=y_2(x_1-y_4)+b-y_1y_3-\psi_1\psi_2,\label{cor-sys-uv-g}\\
&u_2(b-v_1v_3-x_2v_4-\eta_1\eta_2)-v_3(a-u_1y_3-u_2u_4-\xi_1\psi_2)-u_1v_4-\xi_1\eta_2=\label{cor-sys-uv-h}\\
&y_2(a-x_1x_3-x_2x_4-\chi_1\chi_2)-x_3(b-y_1y_3-y_2y_4-\psi_1\psi_2)-y_1x_4-\psi_1\chi_2,\nonumber
\end{align}
\end{subequations}
for the rest of the variables $u_1$, $u_2$, $u_4$, $\xi_1$, $v_3$ and $\xi_1$.

From \eqref{cor-sys-uv-d} we obtain $\xi_1\psi_2=\psi_1\psi_2+(y_2-u_2)\chi_1\psi_2$. Substituting the latter to \eqref{cor-sys-uv-g} and using \eqref{eq-u1-v4}, we obtain $u_2$ as given in \eqref{cor-G-b}. With use of $u_2$, \eqref{cor-sys-uv-a} and \eqref{cor-sys-uv-c} imply $v_3$ and $u_1$ as given in \eqref{cor2-G-b} and \eqref{cor-G-a}, respectively. Subsequently, with use of \eqref{cor2-G-b}, from equation \eqref{cor-sys-uv-b} follows that $v_4$ is given by \eqref{cor2-G-c}, whereas equations \eqref{cor-sys-uv-d} and \eqref{cor-sys-uv-b} imply the following expressions
\begin{equation*}
 \xi_1=\psi_1+\frac{a-b}{x_1-y_4+x_2y_3+\chi_1\psi_2}\chi_1, \qquad \eta_2=\chi_2+\frac{a-b}{x_1-y_4+x_2y_3+\chi_1\psi_2}\psi_2
\end{equation*}
for $\xi_1$ and $\eta_2$. Multiplying both the numerator and the denominator of the fractions in the above equations by the conjugate expression of the denominator, it follows that $\xi_1$ and $\eta_2$ are given by \eqref{cor-G-d} and \eqref{cor2-G-e}, respectively.
\end{proof}

\begin{theorem}\label{YBmap}
Map 
\begin{equation}\label{map-YB-BSQ-G}
S_{a,b}:((\textbf{x},\pmb{\chi}),(\textbf{y},\pmb{\psi}))\rightarrow ((\textbf{u},\pmb{\xi}),(\textbf{v},\pmb{\eta})),
\end{equation}
given by
\begin{subequations}\label{YB-BSQ-G}
\allowdisplaybreaks
\begin{align}
x_1\mapsto u_1&=y_1+\frac{a-b}{x_1-y_4+x_2y_3+\chi_1\psi_2}x_2,\label{YB-BSQ-G-a}\\
x_2\mapsto u_2&=y_2+\frac{b-a}{x_1-y_4+x_2y_3+\chi_1\psi_2},\label{YB-BSQ-G-b}\\
x_3\mapsto u_3&=y_3,\label{YB-BSQ-G-c}\\
x_4\mapsto u_4&=y_4,\label{YB-BSQ-G-d}\\
\chi_1\mapsto \xi_1&=\psi_1+\frac{a-b}{x_1-y_4+x_2y_3}\chi_1,\label{YB-BSQ-G-e}\\
\chi_2\mapsto \xi_2&=\psi_2,\label{YB-BSQ-G-f}\\
y_1\mapsto v_1&=x_1,\label{YB-BSQ-G-g}\\
y_2\mapsto v_2&=x_2,\label{YB-BSQ-G-h}\\
y_3\mapsto v_3&=x_3+\frac{b-a}{x_1-y_4+x_2y_3+\chi_1\psi_2},\label{YB-BSQ-G-i}\\
y_4\mapsto v_4&=x_4+\frac{b-a}{x_1-y_4+x_2y_3+\chi_1\psi_2}y_3,\label{YB-BSQ-G-j}\\
\psi_1\mapsto \eta_1&=\chi_1,\label{YB-BSQ-G-k}\\
\psi_2\mapsto \eta_2&=\chi_2+\frac{a-b}{x_1-y_4+x_2y_3}\psi_2,\label{YB-BSQ-G-l}
\end{align}
\end{subequations}
is a twelve-dimensional parametric Yang-Baxter map, and possesses the following invariants
\begin{subequations}
\begin{align}
I_1&=x_2+y_2-x_3-y_3,\\
I_2&=x_1+y_1+x_2y_2,\\
I_3&=x_4+y_4-x_3y_3,\\
I_4&=b(x_2-x_3)-a(y_3-y_2)+(x_4-x_3y_2-y_1)(x_1+x_2y_3-y_4),
\intertext{as well as the following anti-invariants}
I_5&=\chi_1\psi_1,\qquad\text{and}\quad I_6=\chi_2\psi_2.
\end{align}
\end{subequations}

Moreover, the bosonic limit of \eqref{map-YB-BSQ-G}-\eqref{YB-BSQ-G} is map \eqref{YB-BSQ}-\eqref{YB-Boussinesq}.
\end{theorem}
\begin{proof}
See Appendix \ref{Appendix-A}.
\end{proof}

\begin{corollary}\label{cor-Lax}
Map \eqref{map-YB-BSQ-G}-\eqref{YB-BSQ-G} satisfies the following matrix refactorisation problem:
\begin{equation}\label{Lax-YB-G}
\mathcal{L}_a(u_1,u_2,y_3,y_4,\xi_1,\psi_2)\mathcal{L}_b(x_1,x_2,v_3,v_4,\chi_1,\eta_2)=\mathcal{L}_b(y_1,y_2,y_3,y_4,\psi_1,\psi_2)\mathcal{L}_a(x_1,x_2,x_3,x_4,\chi_1,\chi_2),
\end{equation}
where $\mathcal{L}_a(x_1,x_2,x_3,x_4,\chi_1,\chi_2)$ is given by \eqref{Lax-BSQ-YB-G}.
\end{corollary}

\subsection{Step \textit{III}: Squeeze down to Grassmann extension of the Boussinesq lattice equation}
Here, we shall construct our Grassmann extended system of lattice equations using the symmetries of the Yang-Baxter map which was constructed in the previous section. That said, using the observation that, in \eqref{YB-BSQ-G}, $y_2=x_3$ implies $u_2=v_3$, we shall construct a Grassmann extension of the Boussinesq lattice equation \eqref{BSQ-system}. In particular, we have the following.

\begin{theorem}(Grassmann extension of the Boussinesq lattice system)
Map \eqref{map-YB-BSQ-G}-\eqref{YB-BSQ-G} can be squeezed down to the following system
\begin{align}\label{BSQ-quad-G}
&(p_{01}-p_{10})(p+qq_{11}-r_{11}+\tau\theta_{11})=(a-b)q,\nonumber\\
&(q_{01}-q_{10})(p+qq_{11}-r_{11}+\tau\theta_{11})=b-a,\nonumber\\
&(r_{01}-r_{10})(p+qq_{11}-r_{11}+\tau\theta_{11})=(b-a)q_{11},\\
&(\tau_{01}-\tau_{10})(p+qq_{11}-r_{11}+\tau\theta_{11})=(a-b)\tau,\nonumber\\
&(\theta_{01}-\theta_{10})(p+qq_{11}-r_{11}+\tau\theta_{11})=(a-b)\theta_{11}.\nonumber
\end{align}
System \eqref{BSQ-quad-G} is integrable with Lax representation 
\begin{equation}\label{Laxrepr}
\mathcal{L}_a(p_{01},q_{01},q_{11},r_{11},\tau_{01},\theta_{11})\mathcal{L}_b(p,q,q_{01},r_{01},\tau,\theta_{01})=\mathcal{L}_b(p_{10},q_{10},q_{11},r_{11},\tau_{10},\theta_{11})\mathcal{L}_a(p,q,q_{10},r_{10},\tau,\theta_{10}),
\end{equation}
where
\begin{equation}\label{Lax-BSQ-G}
\mathcal{L}_a(p,q,q_{10},r_{10},\tau,\theta_{10}):=
\left(
\begin{matrix}
 -q_{10} & 1 & 0 & 0 \\
 -r_{10} & 0 & 1 & 0 \\
a-pq_{10}-qr_{10}-\tau\theta_{10}-\lambda & p & q & \tau\\
-\theta_{10} & 0 & 0 & 1
\end{matrix}\right),
\end{equation}
where $p,q,r\in G_0$ and $\tau,\theta\in G_1$.
Furthermore, system \eqref{BSQ-quad-G} possesses the following conservation law
\begin{equation}\label{consLaws-a}
(\mathcal{T}-1)(p_{10}+qq_{10}-r)=(\mathcal{S}-1)(p_{01}+qq_{01}-r),
\end{equation}
and satisfies the following
\begin{equation}\label{consLaws-b}
(\mathcal{T}+1)(\tau\theta_{10})=(\mathcal{S}+1)(\tau\theta_{01}).
\end{equation}
Finally, the bosonic limit of \eqref{BSQ-quad-G} is the Boussinesq lattice equation \eqref{BSQ-system}.
\end{theorem}
\begin{proof}
For map \eqref{map-YB-BSQ-G}-\eqref{YB-BSQ-G}, $y_2=x_3$ implies $u_2=v_3$. Now, relabeling $y_2=x_3=q_{10}$, $u_2=v_3=q_{01}$, $x_1=v_1=p$, $x_2=v_2=q$, $x_4=r_{10}$, $y_3=u_3=q_{11}$, $y_4=u_4=r_{11}$, $\chi_1=\eta_1=\tau$, $\chi_2=\theta_{10}$, $\psi_1=\tau_{10}$, $\psi_2=\xi_2=\theta_{11}$, equations \eqref{YB-BSQ-G-a}, \eqref{YB-BSQ-G-b}, \eqref{YB-BSQ-G-i}, \eqref{YB-BSQ-G-e} and \eqref{YB-BSQ-G-l} imply
\begin{subequations}\label{sys-stair}
\begin{align}
&p_{01}=p_{10}+\frac{a-b}{p-r_{11}+qq_{11}+\tau\theta_{11}}q,\label{sys-stair-a}\\
&q_{01}=q_{10}+\frac{b-a}{p-r_{11}+qq_{11}+\tau\theta_{11}},\label{sys-stair-b}\\
&r_{01}=r_{10}+\frac{b-a}{p-r_{11}+qq_{11}+\tau\theta_{11}}q_{11},\label{sys-stair-c}\\
&\theta_{01}=\theta_{10}+\frac{a-b}{p-r_{11}+qq_{11}}\theta_{11},\label{sys-stair-d}\\
&\tau_{01}=\tau_{10}+\frac{a-b}{p-r_{11}+qq_{11}}\tau,\label{sys-stair-e}
\end{align}
\end{subequations}
which can be rewritten in the form of system \eqref{BSQ-quad-G}. Its Lax representation follows from corollary \eqref{cor-Lax} after the above relabeling.

From equations \eqref{sys-stair-a} and \eqref{sys-stair-c} follows that
\begin{equation*}
p_{01}-p_{10}+q(q_{01}-q_{10})=r_{01}-r_{10}-q_{11}(q_{01}-q_{10})=0.
\end{equation*}
The latter equations imply
\begin{equation*}
p_{01}-p_{10}+q(q_{01}-q_{10})=r_{01}-r_{10}+q_{11}(q_{10}-q_{01}),
\end{equation*}
which is equivalent to \eqref{consLaws-a}.

Moreover, by straightforward calculation one can verify that \eqref{consLaws-b}, namely equation
\begin{equation*}
(\tau_{01}-\tau_{10})\theta_{11}=-\tau(\theta_{01}-\theta_{10}),
\end{equation*}
is identically satisfied in view of equations \eqref{sys-stair-d} and \eqref{sys-stair-e}.

Finally, setting all the odd variables and their shifts equal to zero, namely $\tau=\tau_{10}=\tau_{01}=\theta=\theta_{11}=0$, system \eqref{BSQ-quad-G} implies the Boussinesq lattice equation \eqref{BSQ-system}.
\end{proof}

\begin{remark}\normalfont
Equation \eqref{consLaws-b} can be written in the form of conservation law under the change of variables $\tau \rightarrow (-1)^n\tau$ and $\theta \rightarrow (-1)^{m-1}\theta $, i.e. $\theta_{01} \rightarrow (-1)^m\theta_{01}$.
\end{remark}

For system \eqref{sys-stair} we can set the initial value problem on the staircase, as in figure \ref{fig-ivp}.

\begin{figure}[ht]
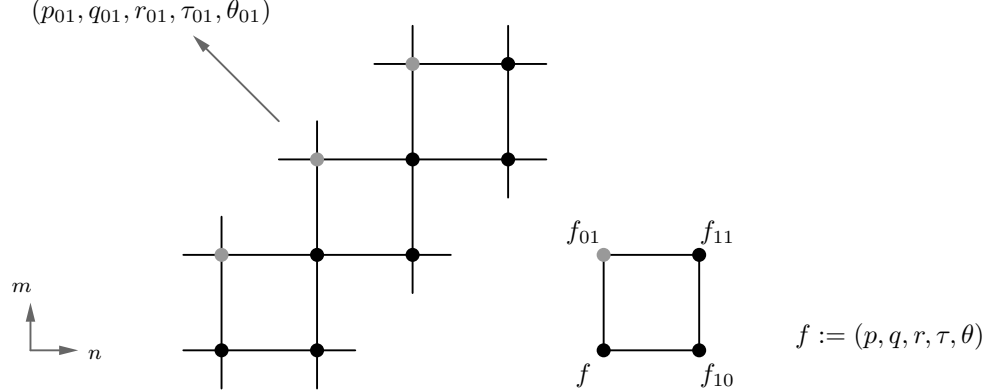

\centertexdraw{
\setunitscale 0.5

\move (0 0) \lvec (1 0) \lvec (1 1) \lvec (2 1) \lvec (2 2) \lvec (3 2) \lvec (3 3)
\move (0 0) \lvec (0 1) \lvec (1 1) \lvec (1 2) \lvec (2 2) \lvec (2 3) \lvec (3 3) 
\move (4 0) \lvec (5 0) \lvec (5 1) \lvec (4 1) \lvec (4 0)
\move (-.4 0) \lvec (0 0)
\move (0 -.4) \lvec (0 0)
\move (3 3) \lvec (3 3.4)
\move (3 3) \lvec (3.4 3)
\move (0 1) \lvec (-.4 1)
\move (0 1) \lvec (0 1.4)
\move (1 2) \lvec (0.6 2)
\move (1 2) \lvec (1 2.4)
\move (2 3) \lvec (1.6 3)
\move (2 3) \lvec (2 3.4)
\move (1 0) \lvec (1.4 0)
\move (1 0) \lvec (1 -.4)
\move (2 1) \lvec (2.4 1)
\move (2 1) \lvec (2 0.6)
\move (3 2) \lvec (3.4 2)
\move (3 2) \lvec (3 1.6)
\move (0 0) \fcir f:0.0 r:0.075
\move (1 0) \fcir f:0.0 r:0.075
\move (1 1) \fcir f:0.0 r:0.075
\move (2 1) \fcir f:0.0 r:0.075
\move (2 2) \fcir f:0.0 r:0.075
\move (3 2) \fcir f:0.0 r:0.075
\move (3 3) \fcir f:0.0 r:0.075
\move (0 1) \fcir f:0.6 r:0.075
\move (1 2) \fcir f:0.6 r:0.075
\move (2 3) \fcir f:0.6 r:0.075
\move (4 0) \fcir f:0.0 r:0.075
\move (5 0) \fcir f:0.0 r:0.075
\move (5 1) \fcir f:0.0 r:0.075
\move (4 1) \fcir f:0.6 r:0.075
\arrowheadsize l:.20 w:.10
\move(-2 0) \linewd 0.02 \setgray 0.4 \arrowheadtype t:F \avec(-1.5 0) 
\move(-2 0) \linewd 0.02 \setgray 0.4 \arrowheadtype t:F \avec(-2 0.5) 
\htext (-1.4 -.1) {\scriptsize{$n$}}
\htext (-2.2 .6) {\scriptsize{$m$}}

\htext (3.7 -.4) {\small{$f$}}
\htext (3.6 1.1) {\small{$f_{01}$}}
\htext (5 -.4) {\small{$f_{10}$}}
\htext (5 1.1) {\small{$f_{11}$}}

\htext (6 0) {\small{$f:=(p,q,r,\tau,\theta)$}}

\htext (-2 3.4) {{\small{$(p_{01},q_{01},r_{01},\tau_{01},\theta_{01})$}}}

\move(.6 2.4) \linewd 0.02 \setgray 0.4 \arrowheadtype t:F \avec(-.3 3.3) 
}
\caption{{\em{Initial value problem on the vertices of the staircase and direction of evolution.}}}
\label{fig-ivp}
\end{figure}

\section{3D consistency of a Grassmann extended Boussinesq-type system}\label{3D-consistency}
Now, conservation law \eqref{consLaws-a} indicates to seek a function $f=f(n,m)$ such that
\begin{subequations}\label{conl-p-q}
\begin{align}
p_{10}+qq_{10}-r&=(\mathcal{S}-1)f,\label{conl-p-q-a}\\
p_{01}+qq_{01}-r&=(\mathcal{T}-1)f,\label{conl-p-q-b}
\end{align}
\end{subequations}
namely, seek $f=f(n,m)$ satisfying the following system difference equations
\begin{align}
f_{10}-f&=p_{10}+qq_{10}-r,\\
f_{01}-f&=p_{01}+qq_{01}-r.
\end{align}
The above imply that
\begin{equation}
f_{01}-f_{10}=p_{01}-p_{10}+q(q_{01}-q_{10})=0.
\end{equation}
The above equation implies $f=C(n+m)$. We restrict ourselves to the case where $C(n+m)=\text{const.}=0$, and from system \eqref{conl-p-q} follows that
\begin{subequations}\label{aux-sys}
\begin{align}
p_{10}+qq_{10}-r&=0,\label{aux-sys-a}\\
p_{01}+qq_{01}-r&=0,\label{aux-sys-b}
\end{align}
\end{subequations}
which are equivalent to \eqref{sys-stair-a} and \eqref{sys-stair-c}.

With the use of \eqref{aux-sys}, we prove the following.

\begin{proposition}
System \eqref{sys-stair} can be written in the form of the following Grassmann extended Boussinesq-type system
\begin{subequations}\label{sys-stair-equiv}
\begin{align}
p_{11}&=\frac{r_{10}q_{01}-r_{01}q_{10}}{q_{01}-q_{10}},\label{sys-stair-equiv-a}\\
q_{11}&=\frac{r_{01}-r_{10}}{q_{01}-q_{10}},\label{sys-stair-equiv-b}\\
r_{11}&=\frac{b-a+q(r_{01}-r_{10})+\tau(\theta_{01}-\theta_{10})+p(q_{01}-q_{10})}{q_{01}-q_{10}},\label{sys-stair-equiv-c}\\
\theta_{11}&=\frac{\theta_{01}-\theta_{10}}{q_{01}-q_{10}},\label{sys-stair-equiv-d}\\
\tau &=\frac{\tau_{01}-\tau_{10}}{q_{01}-q_{10}}\label{sys-stair-equiv-e}
\end{align}
\end{subequations}
where $p,q,r\in G_0$ and $\tau,\theta\in G_1$.
\end{proposition}
\begin{proof}
Shifting equations \eqref{aux-sys-a} and \eqref{aux-sys-b} in the $m$ and $n$ direction, respectively, we obtain
\begin{subequations}\label{p-q-shift}
\begin{align}
p_{11}+q_{01}q_{11}-r_{01}&=0,\label{p-q-shift-a}\\
p_{11}+q_{10}q_{11}-r_{10}&=0.\label{p-q-shift-b}
\end{align}
\end{subequations}
Subtraction of the above and solving for $q_{11}$ implies \eqref{sys-stair-equiv-b}. Now, using the latter, we obtain $p_{11}$ given by \eqref{sys-stair-equiv-a}. Finally, with the use of \eqref{sys-stair-equiv-b}, equation \eqref{sys-stair-e} can be rewritten in the form \eqref{sys-stair-equiv-e}, and using this expression for $\theta_{11}$, \eqref{sys-stair-equiv-a} implies \eqref{sys-stair-equiv-c}.
\end{proof}

\begin{remark}\normalfont
The bosonic limit of system \eqref{sys-stair-equiv} is the Boussinesq lattice system as it appears in \cite{Bridgman}.
\end{remark}

Although system \eqref{sys-stair-equiv} is integrable in the sense that it possesses Lax representation, we cannot claim integrability in the sense of $3D$-consistency, since the term ``$\tau_{11}$'' is missing. Since one of the main purposes of this paper is to prove that the $3D$-consistency is preserved in the Grassmann extension of some systems, we shall prove this property for the bosonic limit of system of system \eqref{sys-stair-equiv}, as $\tau \rightarrow 0$. In particular, we have the following.

\begin{theorem}\label{3D-BSQ}
The system
\begin{subequations}\label{sys-stair-equiv-t0}
\begin{align}
p_{11}&=\frac{r_{10}q_{01}-r_{01}q_{10}}{q_{01}-q_{10}},\label{sys-stair-equiv-t0-a}\\
q_{11}&=\frac{r_{01}-r_{10}}{q_{01}-q_{10}},\label{sys-stair-equiv-t0-b}\\
r_{11}&=\frac{b-a+q(r_{01}-r_{10})+p(q_{01}-q_{10})}{q_{01}-q_{10}},\label{sys-stair-equiv-t0-c}\\
\theta_{11}&=\frac{\theta_{01}-\theta_{10}}{q_{01}-q_{10}},\label{sys-stair-equiv-t0-d}
\end{align}
\end{subequations}
where $p,q,r\in G_0$ and $\theta\in G_1$, has the $3D$-consistency property. 
\end{theorem}

The proof of this theorem is presented in Appendix B. It is worth mentioning that, since the ``111'' values depend on the initial $p$ and $q$ (see Appendix B), the system \eqref{3D-BSQ} does not have the ``tetrahedron property''\footnote{The name of the property is due to the fact that the values `100', `010', `001' and `111' form a tetrahedron (see figure \ref{quad-cube}).}. For the proof the following Lemma is needed.

\begin{lemma}\label{inv-A}
The following function
\begin{equation}
\mathcal{A}(\mathfrak{a}_{100},\mathfrak{a}_{010},\mathfrak{a}_{001},\mathfrak{b}_{100},\mathfrak{b}_{010},\mathfrak{b}_{001})=(\mathfrak{a}_{001}-\mathfrak{a}_{010})(\mathfrak{b}_{001}-\mathfrak{b}_{100})-(\mathfrak{a}_{001}-\mathfrak{a}_{100})(\mathfrak{b}_{001}-\mathfrak{b}_{010}),
\end{equation}
where $\mathfrak{a}$ and $\mathfrak{b}$ can be either odd or even variables, is invariant under simultaneous cyclic permutations of $(\mathfrak{a}_{100},\mathfrak{a}_{010},\mathfrak{a}_{001})$ and $(\mathfrak{b}_{100},\mathfrak{b}_{010},\mathfrak{b}_{001})$.
\end{lemma}
\begin{proof}
It is 
$$\mathcal{A}(\mathfrak{a}_{100},\mathfrak{a}_{010},\mathfrak{a}_{001},\mathfrak{b}_{100},\mathfrak{b}_{010},\mathfrak{b}_{001})=\mathfrak{b}_{100}(\mathfrak{a}_{010}-\mathfrak{a}_{001})+\mathfrak{b}_{010}(\mathfrak{a}_{001}-\mathfrak{a}_{100})+\mathfrak{b}_{001}(\mathfrak{a}_{100}-\mathfrak{a}_{010}).$$ It can be verified by straightforward calculation that 
$$\mathcal{A}(\mathfrak{a}_{100},\mathfrak{a}_{010},\mathfrak{a}_{001},\mathfrak{b}_{100},\mathfrak{b}_{010},\mathfrak{b}_{001})=\mathcal{A}(\mathfrak{a}_{001},\mathfrak{a}_{100},\mathfrak{a}_{010},\mathfrak{b}_{001},\mathfrak{b}_{100},\mathfrak{b}_{010})=\mathcal{A}(\mathfrak{a}_{010},\mathfrak{a}_{001},\mathfrak{a}_{100},\mathfrak{b}_{010},\mathfrak{b}_{001},\mathfrak{b}_{100}).$$
\end{proof}

The Lax representation of system \eqref{3D-BSQ} is given by \eqref{Laxrepr}-\eqref{Lax-BSQ-G} for $\tau\rightarrow 0$, namely by
\begin{equation}\label{Laxrepr-T0}
\mathcal{L}_a(p_{01},q_{01},q_{11},r_{11},\theta_{11})\mathcal{L}_b(p,q,q_{01},r_{01},\theta_{01})=\mathcal{L}_b(p_{10},q_{10},q_{11},r_{11},\theta_{11})\mathcal{L}_a(p,q,q_{10},r_{10},\theta_{10}),
\end{equation}
where
\begin{equation}\label{Lax-BSQ-G-T0}
\mathcal{L}_a(p,q,q_{10},r_{10},\tau,\theta_{10}):=
\left(
\begin{matrix}
 -q_{10} & 1 & 0 & 0 \\
 -r_{10} & 0 & 1 & 0 \\
a-pq_{10}-qr_{10}-\tau\theta_{10}-\lambda & p & q & 0\\
-\theta_{10} & 0 & 0 & 1
\end{matrix}\right).
\end{equation}

\section{Concluding remarks}
In this paper, on one hand, we construct some novel, integrable, noncommutative (Grassmann) Boussinesq type systems, namely systems \eqref{BSQ-quad-G} and \eqref{sys-stair-equiv}, together with a $3D$-consistent limit \eqref{sys-stair-equiv}, namely system \eqref{sys-stair-equiv-t0}. Moreover, we derive a novel Yang-Baxter map \eqref{YB-BSQ} together with its Grassmann extension \eqref{map-YB-BSQ-G}. On the other hand, this paper answers an important question regarding the $3D$ consistency of systems, when they are extended to the Grassmann case. That is, not all systems lose their property to be $3D$ consistent in their noncommutative extension, which can be demonstrated by system \eqref{sys-stair-equiv-t0}. 

The $3D$-consistency property of systems like \eqref{sys-stair-equiv-t0} is a very important, since for systems with such property we can:
\begin{itemize}
\item algorithmically construct its Lax representation \cite{Bobenko-Suris, Bridgman, Hiet-Frank-Joshi, Frank2002};
\item obtain a B\"acklund tranformation \cite{Atkinson2006, Hiet-Frank-Joshi}.
\end{itemize}
In our case the Lax representation is already known.

Regarding the $3D$-consistency of the system \eqref{sys-stair-equiv}, we would like to stress out the following. The demand that quad-graph equations \eqref{eq-quad-graph} need to be linear in every variable is because we need to be able to solve uniquely for any of the fields $f$, $f_{10}$, $f_{01}$ and $f_{11}$. This is essential for the $3D$-consistency property. Nevertheless, this is not quite the case for systems with anticommutative variables: In our system, \eqref{sys-stair-equiv}, all equations are linear in all variables, which is obvious if one rewrites the equations in polynomial form. However, this does not imply unique solvability. For instance, equation \eqref{sys-stair-equiv-c} cannot be solved for neither $\tau$, nor $\theta_{10}$, nor $\theta_{01}$.

Our results can be extended in several ways. We list a couple of problems for future work.
\begin{enumerate}
\item The complete (Liouville) integrability of maps \eqref{map-YB-BSQ-G}-\eqref{YB-BSQ-G} and \eqref{YB-BSQ}-\eqref{YB-Boussinesq} is an open problem. We conjecture that there is a suitable Poisson bracket with respect to which the maps' invariants are in involution.
\item Study the solutions of system \eqref{sys-stair-equiv-t0}. In particular, knowing that system \eqref{sys-stair-equiv-t0} has the $3D$-consistency property, we can derive a B\"acklund transformation by setting $(p_{001},q_{001},r_{001},\theta_{001})\equiv (u,v,w,\phi)$, and rewrite \eqref{sys-front}, \eqref{sys-left} as a B\"acklund tranformation between $(p,q,r,\theta)$ and $(u,v,w,\phi)$,\footnote{If $(p,q,r,\theta)$ satisfy \eqref{sys-stair-equiv-t0}, then so do $(u,v,w,\phi)=(p_{001},q_{001},r_{001},\theta_{001})$.} namely:
\begin{align*}
& (u_{10}-r_{10})v+(w-u_{10})q_{10}=0,\\
& v_{10} (v-q_{10})+r_{10}-w=0,\\
& (w_{10}-p)(v-q_{10})-q(w-r_{10})=c-a,\\
&\phi_{10}(v-q_{10})+\theta_{10}-\phi,
\intertext{and}
& (u_{01}-r_{01})v+(w-u_{01})q_{01}=0,\\
& v_{01} (v-q_{01})+r_{01}-w=0,\\
& (w_{01}-p)(v-q_{01})-q(w-r_{01})=c-a,\\
&\phi_{01}(v-q_{01})+\theta_{01}-\phi.
\end{align*}
\item Continuum limits. Using the above B\"acklund transformation to derive solutions of system  \eqref{sys-stair-equiv-t0} and, then, considering the continuum limits of these solutions, we can study the behaviour of the solutions of the corresponding Boussinesq-type system of PDEs.
\end{enumerate}

\section*{Acknowledgements}
This work was carried out within the framework of the State Programme of the Ministry of Education and Science of the Russian Federation, project No. 1.13560.2019/13.1. I would like to thank Dr. T.E. Kouloukas for numerous, useful discussions. I also acknowledge financial support from the London Mathematical Society (short visit grant LMS scheme-2 Ref. 21717). During my visit to the UK, I benefited from talking with Dr. G. Grahovksi, Dr. G. Papamikos, Dr. P. Adamopoulou and Dr. A. Doikou, so I would like to thank them for all the useful discussions and the hospitality. Many thanks to Dr. P. Kassotakis for some comments on the literature. Finally, I would like to thank the anonymous referee for their comments which helped to revise the text in a more motivating manner.

\appendix
\section{Proof of theorem \ref{YBmap}}\label{Appendix-A}
Throughout this proof we shall be using the following ``tilde-hat'' notation. 
\begin{align*}
S^{23}_{b,c}((\pmb{x},\pmb{\chi}),(\pmb{y},\pmb{\psi}),(\pmb{z},\pmb{\zeta}))&=((\pmb{x},\pmb{\chi}),(\tilde{\pmb{y}},\tilde{\pmb{\psi}}),(\tilde{\pmb{z}},\tilde{\pmb{\zeta}}));\\
S^{13}_{a,c}\circ S^{23}_{b,c}((\pmb{x},\pmb{\chi}),(\pmb{y},\pmb{\psi}),(\pmb{z},\pmb{\zeta}))&=((\tilde{\pmb{x}},\tilde{\pmb{\chi}}),(\tilde{\pmb{y}},\tilde{\pmb{\psi}}),(\tilde{\tilde{\pmb{z}}},\tilde{\tilde{\pmb{\zeta}}}));\\
S^{12}_{a,b}\circ S^{13}_{a,c}\circ S^{23}_{b,c}((\pmb{x},\pmb{\chi}),(\pmb{y},\pmb{\psi}),(\pmb{z},\pmb{\zeta}))&=((\tilde{\tilde{\pmb{x}}},\tilde{\tilde{\pmb{\chi}}}),(\tilde{\tilde{\pmb{y}}},\tilde{\tilde{\pmb{\psi}}}),(\tilde{\tilde{\pmb{z}}},\tilde{\tilde{\pmb{\zeta}}})),
\end{align*}
according to right side of the Yang-Baxter equation. Now, according to the right side of the Yang-Baxter equation, 
\begin{align*}
S^{12}_{a,b}((\pmb{x},\pmb{\chi}),(\pmb{y},\pmb{\psi}),(\pmb{z},\pmb{\zeta}))&=((\hat{\pmb{x}},\hat{\pmb{\chi}}),(\hat{\pmb{y}},\hat{\pmb{\psi}}),(\pmb{z},\pmb{\zeta}));\\
S^{13}_{a,c}\circ S^{12}_{a,b}((\pmb{x},\pmb{\chi}),(\pmb{y},\pmb{\psi}),(\pmb{z},\pmb{\zeta}))&=((\hat{\hat{\pmb{x}}},\hat{\hat{\pmb{\chi}}}),(\hat{\pmb{y}},\hat{\pmb{\psi}}),(\hat{\pmb{z}},\hat{\pmb{\zeta}}));\\
S^{23}_{b,c}\circ S^{13}_{a,c}\circ S^{12}_{a,b}((\pmb{x},\pmb{\chi}),(\pmb{y},\pmb{\psi}),(\pmb{z},\pmb{\zeta}))&=((\hat{\hat{\pmb{x}}},\hat{\hat{\pmb{\chi}}}),(\hat{\hat{\pmb{y}}},\hat{\hat{\pmb{\psi}}}),(\hat{\hat{\pmb{z}}},\hat{\hat{\pmb{\zeta}}})).
\end{align*}

Next, we apply the left part of the Yang-Baxter equation to the product
\begin{align}
\mathcal{L}_c(z_1,z_2,z_3,z_4,&\zeta_1,\zeta_2)\mathcal{L}_b(y_1,y_2,y_3,y_4,\psi_1,\psi_2)\mathcal{L}_a(x_1,x_2,x_3,x_4,\chi_1,\chi_2)=\nonumber\\
&\mathcal{L}_c(z_1,z_2,z_3,z_4,\zeta_1,\zeta_2)\mathcal{L}_a(\tilde{x}_1,\tilde{x}_2,y_3,y_4,\tilde{\chi}_1,\psi_2)\mathcal{L}_b(x_1,x_2,\tilde{y}_3,\tilde{y}_4,\chi_1,\tilde{\psi}_2)=\nonumber\\
&\mathcal{L}_a(\tilde{\tilde{x}}_1,\tilde{\tilde{x}}_2,z_3,z_4,\tilde{\tilde{\chi}}_1,\zeta_2)\mathcal{L}_c(\tilde{x}_1,\tilde{x}_2,\tilde{z}_3,\tilde{z}_4,\tilde{\chi}_1,\tilde{\zeta}_2)\mathcal{L}_b(x_1,x_2,\tilde{y}_3,\tilde{y}_4,\chi_1,\tilde{\psi}_2)=\nonumber\\
&\mathcal{L}_a(\tilde{\tilde{x}}_1,\tilde{\tilde{x}}_2,z_3,z_4,\tilde{\tilde{\chi}}_1,\zeta_2)\mathcal{L}_b(\tilde{x}_1,\tilde{x}_2,\tilde{z}_3,\tilde{z}_4,\tilde{\chi}_1,\tilde{\zeta}_2)\mathcal{L}_c(x_1,x_2,\tilde{\tilde{z}}_3,\tilde{\tilde{z}}_4,\chi_1,\tilde{\tilde{\zeta}}_2),
\end{align}
where we have used \eqref{Lax-YB-G} consecutively. Furthermore, applying the right part of the Yang-Baxter equation to the product on the same product,
\begin{align}
\mathcal{L}_c(z_1,z_2,z_3,z_4,&\zeta_1,\zeta_2)\mathcal{L}_b(y_1,y_2,y_3,y_4,\psi_1,\psi_2)\mathcal{L}_a(x_1,x_2,x_3,x_4,\chi_1,\chi_2)=\nonumber\\
&\mathcal{L}_b(\hat{y}_1,\hat{y}_2,z_3,z_4,\hat{\psi}_1,\zeta_2)\mathcal{L}_c(y_1,y_2,\hat{z}_3,\hat{z}_4,\psi_1,\hat{\zeta}_2)\mathcal{L}_a(x_1,x_2,x_3,x_4,\chi_1,\chi_2)=\nonumber\\
&\mathcal{L}_b(\hat{y}_1,\hat{y}_2,z_3,z_4,\hat{\psi}_1,\zeta_2)\mathcal{L}_a(\hat{x}_1,\hat{x}_2,\hat{z}_3,\hat{z}_4,\hat{\chi}_1,\hat{\zeta}_2)\mathcal{L}_c(x_1,x_2,\hat{\hat{z}}_3,\hat{\hat{z}}_4,\chi_1,\hat{\hat{\zeta}}_2)=\nonumber\\
&\mathcal{L}_a(\hat{\hat{x}}_1,\hat{\hat{x}}_2,z_3,z_4,\hat{\hat{\chi}}_1,\zeta)\mathcal{L}_b(\hat{x}_1,\hat{x}_2,\hat{z}_3,\hat{z}_4,\hat{\chi}_1,\hat{\zeta}_2)\mathcal{L}_c(x_1,x_2,\hat{\hat{z}}_3,\hat{\hat{z}}_4,\chi_1,\hat{\hat{\zeta}}_2).
\end{align}
We need to show that the matrix trifactorisation problem:
\begin{eqnarray}
&\mathcal{L}_a(\tilde{\tilde{x}}_1,\tilde{\tilde{x}}_2,z_3,z_4,\tilde{\tilde{\chi}}_1,\zeta_2)\mathcal{L}_b(\tilde{x}_1,\tilde{x}_2,\tilde{z}_3,\tilde{z}_4,\tilde{\chi}_1,\tilde{\zeta}_2)\mathcal{L}_c(x_1,x_2,\tilde{\tilde{z}}_3,\tilde{\tilde{z}}_4,\chi_1,\tilde{\tilde{\zeta}}_2)&=\nonumber\\
&\mathcal{L}_a(\hat{\hat{x}}_1,\hat{\hat{x}}_2,z_3,z_4,\hat{\hat{\chi}}_1,\zeta_2)\mathcal{L}_b(\hat{x}_1,\hat{x}_2,\hat{z}_3,\hat{z}_4,\hat{\chi}_1,\hat{\zeta}_2)\mathcal{L}_c(x_1,x_2,\hat{\hat{z}}_3,\hat{\hat{z}}_4,\chi_1,\hat{\hat{\zeta}}_2),&\label{A-trifac}
\end{eqnarray}
implies
\begin{equation*}
\tilde{\tilde{x}}_i=\hat{\hat{x}}_i, \quad \tilde{\tilde{\chi}}_1=\hat{\hat{x}}_1, \quad \tilde{x}_i=\hat{x}_i, \quad \tilde{z}_j=\hat{z}_j, \quad \tilde{\chi}_1=\hat{\chi}_1,\quad \tilde{\tilde{z}}_j=\hat{\hat{z}}_j, \quad \tilde{\tilde{\zeta}}_2=\hat{\hat{\zeta}}_2,
\end{equation*}
where $i=1,2$ and $j=3,4$.

Indeed, equation \eqref{A-trifac} yields the following system of equations
\begin{subequations}\label{variety-G}
\allowdisplaybreaks
\begin{align}
&v_4-x_3v_3=y_4-x_3y_3,\label{variety-G-a}\\
&w_3-v_2=z_3-y_2,\label{variety-G-b}\\
&v_3-u_2=y_3-x_2,\label{variety-G-c}\\
&\eta_1+v_2\zeta_1=y_2\zeta_1+\psi_1,\label{variety-G-d}\\
&v_1+v_2\zeta_2=y_1+y_2z_2,\label{variety-G-e}\\
&\eta_2-\chi_2v_3=\psi_2-\chi_2y_3,\label{variety-G-f}\\
&\gamma_2+\chi_2(v_3w_3-w_4)-\eta_2w_3=\zeta_2+\chi_2(y_3z_3-z_4)-\psi_2z_3,\label{variety-G-g}\\
&\xi_1+u_2(v_2\zeta_1+\eta_1)+u_1\zeta_1=x_1\zeta_1+x_2(y_2\zeta_1+\psi_1)+x_1\zeta_1,\label{variety-G-h}\\
&u_1(z_2-x_3)+u_2(v_1-x_4+v_2z_2)-\xi_1\chi_2=x_1(z_2-x_3)+x_2(y_1-x_4+y_2z_2)-\chi_1\chi_2,\label{variety-G-i}\\
&w_3(v_4-z_1-x_3v_3)+w_4(x_3-z_2)-\zeta_1\gamma_2=z_3(y_4-z_1-x_3y_3)+z_4(x_3-z_2)-\zeta_1\zeta_2,\label{variety-G-j}\\
&v_3(x_4-v_1)-v_2(v_4-z_1)-\eta_1\eta_2=y_3(x_4-y_1)-y_2(y_4-z_1)-\psi_1\psi_2,\label{variety-G-k}\\
&w_4-u_1-v_3w_3+u_2w_3-u_2v_2=z_4-x_1-y_3z_3+x_2z_3-x_2y_2,\label{variety-G-l}\\
&w_3\left[(v_1-x_4)v_3+v_2v_4+\eta_1\eta_2-b\right]+v_2(c-z_1w_3-z_2w_4-\zeta_1\gamma_2)+(x_4-v_1)w_4-\eta_1\gamma_2=\nonumber\\
&z_3\left[(y_1-x_4)y_3+y_2y_4+\psi_1\psi_2-b\right]+y_2(c-z_1z_3-z_2z_4-\zeta_1\zeta_2)+(x_4-y_1)z_4-\psi_1\zeta_2,\label{variety-G-m}\\
&u_1(z_1-v_4)+u_2\left[b-v_1v_3+(z_1-v_4)v_2-\eta_1\eta_2\right]-v_3(a-u_1x_3-u_2x_4-\xi_1\chi_2)-\xi_1\eta_2=\nonumber\\
&x_1(z_1-y_4)+x_2\left[b-y_1y_3(z_1-y_4)y_2-\psi_1\psi_2\right]-y_3(a-x_1x_3-x_2x_4-\chi_1\chi_2)-\chi_1\psi_2,\label{variety-G-n}\\
&u_1\left[c+(v_4-z_1)w_3-\zeta_1\gamma_2\right]+(a-u_1x_3-u_2x_4-\xi_1\chi_2)(v_3w_3-w_4)+\xi_1(\eta_2w_3-\gamma_2)+\nonumber\\
&u_2\left[w_3(v_1v_3+v_2v_4+\eta_1\eta_2-b)-v_1w_4+v_2(c-z_1z_3-z_2w_4-\zeta_1\gamma_2)-\eta_1\gamma_2\right]=\nonumber\\
&x_1\left[c+(y_4-z_1)z_3-z_2z_4-\zeta_1\zeta_2\right]+(a-x_1x_3-x_2x_4-\chi_1\chi_2)(y_3z_3-z_4)+\chi_1(\psi_2z_3-\zeta_2)+\nonumber\\
&x_2\left[z_3(y_1y_4+y_2y_4+\psi_1\psi_2-b)-y_1z_4+y_2(c-z_1z_3-z_2z_4-\zeta_1\zeta_2)-\psi_1\psi_2\right].\label{variety-G-o}
\end{align}
\end{subequations}

Using \eqref{variety-G-a}-\eqref{variety-G-f}, we express in terms of ``$u_1-x_1$'', ``$u_2-x_2$'' and ``$v_2-y_2$'' all variables $v_1$, $v_3$, $v_4$, $w_3$, $\eta_1$ and $\eta_2$, namely
\begin{subequations}\label{variety-G'}
\begin{align}
v_1&=y_1-(v_2-y_2)z_2,\label{variety-G'-a}\\
v_3&=y_3+u_2-x_2,\label{variety-G'-b}\\
v_4&=y_4+x_3(u_2-x_2),\label{variety-G'-c}\\
w_3&=z_3+v_2-y_2,\label{variety-G'-d}\\
\eta_1&=y_1-(v_2-y_2)z_2,\label{variety-G'-e}\\
\eta_2&=\psi_2+(u_2-x_2)\chi_2.\label{variety-G'-f}
\end{align}
\end{subequations}
Now, relation \eqref{variety-G-l}, using \eqref{variety-G'-b} and \eqref{variety-G'-d}, implies
\begin{equation}\label{w4}
w_4=z_4+u_1-x_1+y_3(v_2-y_2)+v_2(u_2-x_2),
\end{equation}
whereas from \eqref{variety-G-h}, with use of \eqref{variety-G'-e}, follows that
\begin{equation}\label{xi}
\xi_1=\chi_1+(x_1-u_1)\zeta_1+(x_2-u_2)(\psi_1+y_2\zeta_1).
\end{equation}
Moreover, from \eqref{variety-G-g}, in view of \eqref{variety-G'-d} and \eqref{variety-G-f}, we obtain
\begin{equation}\label{gamma2}
\gamma_2=\zeta_2+\chi_2(u_1-x_1)+v_2(u_1-x_2)\chi_2+\psi_2(v_2-y_2),
\end{equation}
where we have made use of \eqref{w4}. Additionally, with use of \eqref{variety-G'-a} and \eqref{xi}, equation \eqref{variety-G-i} implies an expression for $u_1$ in terms of ``$u_2-x_2$'', namely
\begin{equation}\label{u1}
u_1=x_1-(u_2-x_2)\frac{y_1+y_2z_2-x_4+(\psi_1+y_2\zeta_1)\chi_2}{z_2-x_3+\zeta_1\chi_2}.
\end{equation}
Additionally, with the help of \eqref{variety-G'-a}, \eqref{variety-G'-b}, \eqref{variety-G'-c}, \eqref{variety-G'-e} and \eqref{variety-G'-f}, it follows from \eqref{variety-G-k} that $v_2$ can be expressed as
\begin{equation}\label{v2}
v_2=y_2+(u_2-x_2)\frac{y_1+y_2x_3-x_4+\psi_1\chi_2}{z_1+z_2y_3-y_4+\zeta_1\psi_2+(u_2-x_2)(z_2-x_3+\zeta_1\chi_2)}.
\end{equation}

Equation \eqref{variety-G-n} can be rewritten in the form
\begin{align}\label{auxil}
&a(y_3-v_3)+b(u_2-x_2)+v_3\left[u_1x_3+u_2x_4+\xi_1\chi_2\right]-u_2\left[v_1v_3+v_2(v_4-z_1)+\eta_1\eta_2\right]-\xi_1\eta_2=\nonumber\\
&y_3(x_1x_3+x_2x_4+\chi_1\chi_2)-x_2\left[y_1y_3+y_2(y_4-z_1)+\psi_1\psi_2\right]+x_1(z_1-y_4)-\chi_1\psi_2.
\end{align}
Using equations \eqref{variety-G-i} and \eqref{variety-G-k}, the quantities in square brackets in the left-hand side part of the above equation, can be substituted by the following expressions
\begin{subequations}\label{exprSB}
\begin{align}
&u_1x_3+u_2x_4+\xi_1\chi_2=(u_1-x_1)z_2+(u_2-x_2)(y_1+y_2z_2)+x_1x_3+x_2x_4+\chi_1\chi_2,\\
&v_1v_3+v_2(v_4-z_1)+\eta_1\eta_2=x_4(v_3-y_3)+y_1y_3+y_2y_4-y_2z_1+\psi_1\psi_2,
\end{align}
\end{subequations}
where we have used \eqref{variety-G-e}. 

After a little manipulation, equation \eqref{auxil}, with use of equations \eqref{exprSB}, can be written as
\begin{subequations}\label{expression}
\begin{align}
&(u_2-x_2)\left[b-a+y_2(y_3z_2-y_4+z_1+\zeta_1\psi_2)\right]+(u_1-x_1)(z_1-y_4+y_3z_2+\zeta_1\psi_2)+\nonumber\\
&(u_1-x_1)(u_2-x_2)(z_2-x_3+\zeta_1\chi_2)+(u_2-x_2)^2\left[y_1-x_4+y_2z_2+(\psi_1+y_2\zeta_1)\chi_2\right].
\end{align}
\end{subequations}
But, due to \eqref{u1}, $(u_1-x_1)(u_2-x_2)(z_2-x_3+\zeta_1\chi_2)=-(u_2-x_2)^2(y_1+y_2z_2-x_4+\left[\psi_1+y_2\zeta_1)\chi_2\right]$, and with this observation, equation \eqref{expression} can be factorised as
\begin{equation}
(u_2-x_2)\frac{(b-a)(z_2-x_3+\zeta_1\chi_2)-(y_1+y_2x_3-x_4+\psi_1\chi_2)(z_1+y_3z_2-y_4+\zeta_1\psi_2)}{z_2-x_3+\zeta_1\chi_2}=0,
\end{equation}
which implies $u_2=x_2$.

With $u_2=x_2$, we obtain
\begin{equation*}
u_1=x_1,\quad v_2=y_2,\quad v_3=y_3,\quad v_4=y_4,\quad \text{and} ~~~\eta_2=\psi_2,
\end{equation*}
from \eqref{u1}, \eqref{v2}, \eqref{variety-G'-b}, \eqref{variety-G'-c} and \eqref{variety-G'-f}, respectively, and using the above, it follows that
\begin{equation*}
v_1=y_1,\quad w_3=z_3,\quad w_4=z_4,\quad \xi_1=\chi_1, \quad \eta_1=\psi_1 \quad \text{and} ~~~ \gamma_2=z_2,
\end{equation*}
in view of \eqref{variety-G'-a}, \eqref{variety-G'-d}, \eqref{w4}, \eqref{xi}, \eqref{variety-G'-e} and \eqref{gamma2}.

Now, map \eqref{map-YB-BSQ-G}-\eqref{YB-BSQ-G} shares the same invariants $I_i$, $i=1,\ldots,4$ as  \eqref{YB-BSQ}-\eqref{YB-Boussinesq}, which can be verified by straightforward calculation. Moreover,
\begin{equation}
\xi_i\eta_i = \psi_i\chi_i=-\chi_i\psi_i,~~ i=1,2,
\end{equation}
namely the quantities $\xi_i\eta_i$, $i=1,2$, constitute anti-invariants of the map.

Finally, the bosonic limit can be calculated by substituting $\chi_i \rightarrow 0, \psi_i \rightarrow 0$, $i=1,2$, to \eqref{YB-BSQ-G}, and the result will be map \eqref{YB-BSQ}-\eqref{YB-Boussinesq}.

\section{Proof of theorem \ref{3D-BSQ}}\label{Appendix-B}
We write system \eqref{sys-stair-equiv-t0} on the bottom face of the cube in Figure \ref{sqr-cube}, namely
\begin{align}
p_{110}&=\frac{r_{100}q_{010}-r_{010}q_{100}}{q_{010}-q_{100}},\nonumber\\
q_{110}&=\frac{r_{010}-r_{100}}{q_{010}-q_{100}},\nonumber\\
r_{110}&=\frac{b-a+q(r_{010}-r_{100})+p(q_{010}-q_{100})}{q_{010}-q_{100}},\label{sys-btm}\\
\theta_{110}&=\frac{\theta_{010}-\theta_{100}}{q_{010}-q_{100}}.\nonumber
\end{align}
Moreover, according to front face of the cube, the system \eqref{sys-stair-equiv-t0} is written as
\begin{align}
p_{101}&=\frac{r_{100}q_{001}-r_{001}q_{100}}{q_{001}-q_{100}},\nonumber\\
q_{101}&=\frac{r_{001}-r_{100}}{q_{001}-q_{100}},\nonumber\\
r_{101}&=\frac{c-a+q(r_{001}-r_{100})+p(q_{001}-q_{100})}{q_{001}-q_{100}},\label{sys-front}\\
\theta_{101}&=\frac{\theta_{001}-\theta_{100}}{q_{001}-q_{100}}.\nonumber
\end{align}
whereas on the left side of the cube is expressed as
\begin{align}
p_{011}&=\frac{r_{010}q_{001}-r_{001}q_{010}}{q_{001}-q_{010}},\nonumber\\
q_{011}&=\frac{r_{001}-r_{010}}{q_{001}-q_{010}},\nonumber\\
r_{011}&=\frac{c-b+q(r_{001}-r_{010})+p(q_{001}-q_{010})}{q_{001}-q_{010}}+p,\label{sys-left}\\
\theta_{011}&=\frac{\theta_{001}-\theta_{010}}{q_{001}-q_{010}}.\nonumber
\end{align}

There are three different ways to obtain the values $(p_{111},q_{111},r_{111},\theta_{111})$: A) Shifting system \eqref{sys-btm} in the $k$-direction, using \eqref{sys-front} and \eqref{sys-left} to replace the ``101'' and ``011'' values, B) shifting system \eqref{sys-front} in the $m$-direction, using \eqref{sys-btm} and \eqref{sys-left} to replace the ``110'' and ``011'' values, and C) shifting \eqref{sys-left} in the $n$-direction, using \eqref{sys-btm} and \eqref{sys-front} to replace the ``110'' and ``101'' values, respectively.

A) The final ``111'' values we obtain are
\begin{subequations}\label{111-from-btm}
\begin{align}
p_{111}&=\frac{p\mathcal{A}(r_{100},r_{010},r_{001},q_{100},q_{010},q_{001})-\mathcal{B}_{a,b,c}(r_{100},r_{010},r_{001})}{\mathcal{A}(r_{100},r_{010},r_{001},q_{100},q_{010},q_{001})},\label{111-from-btm-a}\\
q_{111}&=\frac{q\mathcal{A}(r_{100},r_{010},r_{001},q_{100},q_{010},q_{001})+\mathcal{B}_{a,b,c}(q_{100},q_{010},q_{001})}{\mathcal{A}(r_{100},r_{010},r_{001},q_{100},q_{010},q_{001})},\label{111-from-btm-b}\\
r_{111}&=\frac{(p_{100}+qq_{100})\mathcal{A}(r_{100},r_{010},r_{001},q_{100},q_{010},q_{001})+\mathcal{C}_{a,b,c}(q_{100},q_{010},q_{001})}{\mathcal{A}(r_{100},r_{010},r_{001},q_{100},q_{010},q_{001})},\label{111-from-btm-c}\\
\theta_{111}&=\frac{\mathcal{A}(\theta_{100},\theta_{010},\theta_{001},q_{100},q_{010},q_{001})}{\mathcal{A}(r_{100},r_{010},r_{001},q_{100},q_{010},q_{001})},\label{111-from-btm-d}
\end{align}
\end{subequations}
where 
\begin{align*}
&\mathcal{A}(x_{100},x_{010},x_{001},q_{100},q_{010},q_{001}):=(x_{001}-x_{010})(q_{001}-q_{100})-(x_{001}-x_{100})(q_{001}-q_{010}),\\
&\mathcal{B}_{a,b,c}(x_{100},x_{010},x_{001}):=a(x_{001}-x_{010})+b(x_{100}-x_{001})+c(x_{010}-x_{100}),\\
&\mathcal{C}_{a,b,c}(x_{100},x_{010},x_{001}):=ax_{100}(x_{001}-x_{010})+bx_{010}(x_{100}-x_{001})+cx_{001}(x_{010}-x_{100}).
\end{align*}
In the derivation of $r_{111}$  in \eqref{111-from-btm-c} we used the relation $p_{001}=p_{100}+q(q_{100}-q_{001})$, which is derived from the system \eqref{aux-sys} written on the front side of the cube, namely the system
\begin{align*}
p_{100}&=r-qq_{100},\\
p_{010}&=r-qq_{010}.
\end{align*}

B) In this case, the final ``111'' values read
\begin{subequations}\label{111-from-front}
\begin{align}
p_{111}&=\frac{p\mathcal{A}(r_{001},r_{100},r_{010},q_{001},q_{100},q_{010})-\mathcal{B}_{a,b,c}(r_{100},r_{010},r_{001})}{\mathcal{A}(r_{001},r_{100},r_{010},q_{001},q_{100},q_{010})},\label{111-from-front-a}\\
q_{111}&=\frac{q\mathcal{A}(r_{001},r_{100},r_{010},q_{001},q_{100},q_{010})+\mathcal{B}_{a,b,c}(q_{100},q_{010},q_{001})}{\mathcal{A}(r_{001},r_{100},r_{010},q_{001},q_{100},q_{010})},\label{111-from-front-b}\\
r_{111}&=\frac{(p_{100}+qq_{100})\mathcal{A}(r_{001},r_{100},r_{010},q_{001},q_{100},q_{010})+\mathcal{C}_{a,b,c}(q_{100},q_{010},q_{001})}{\mathcal{A}(r_{001},r_{100},r_{010},q_{001},q_{100},q_{010})},\label{111-from-front-c}\\
\theta_{111}&=\frac{\mathcal{A}(\theta_{001},\theta_{100},\theta_{010},q_{001},q_{100},q_{010})}{\mathcal{A}(r_{001},r_{100},r_{010},q_{001},q_{100},q_{010})},\label{111-from-front-d}
\end{align}
\end{subequations}
where we have used the relation $p_{010}=p_{100}+q(q_{100}-q_{010})$, derived from the system 
\begin{align*}
p_{100}&=r-qq_{100},\\
p_{010}&=r-qq_{010},
\end{align*}
i.e. system \eqref{aux-sys} expressed on the bottom side of the cube.

C) Finally, the ``111'' values in this case are given by
\begin{subequations}\label{111-from-left}
\begin{align}
p_{111}&=\frac{p\mathcal{A}(r_{010},r_{001},r_{100},q_{010},q_{001},q_{100})-\mathcal{B}_{a,b,c}(r_{100},r_{010},r_{001})}{\mathcal{A}(r_{010},r_{001},r_{100},q_{010},q_{001},q_{100})},\label{111-from-left-a}\\
q_{111}&=\frac{q\mathcal{A}(r_{010},r_{001},r_{100},q_{010},q_{001},q_{100})+\mathcal{B}_{a,b,c}(q_{100},q_{010},q_{001})}{\mathcal{A}(r_{010},r_{001},r_{100},q_{010},q_{001},q_{100})},\label{111-from-left-b}\\
r_{111}&=\frac{(p_{100}+qq_{100})\mathcal{A}(r_{010},r_{001},r_{100},q_{010},q_{001},q_{100})+\mathcal{C}_{a,b,c}(q_{100},q_{010},q_{001})}{\mathcal{A}(r_{010},r_{001},r_{100},q_{010},q_{001},q_{100})},\label{111-from-left-c}\\
\theta_{111}&=\frac{\mathcal{A}(\theta_{010},\theta_{001},q_{010},\theta_{100},q_{001},q_{100})}{\mathcal{A}(r_{010},r_{001},r_{100},q_{010},q_{001},q_{100})}.\label{111-from-left-d}
\end{align}
\end{subequations}
The values \eqref{111-from-btm}, \eqref{111-from-front} and \eqref{111-from-left} coincide due to Lemma \ref{inv-A}.


\begin{thebibliography}{10}

\bibitem{ABS-2005}
{\sc Adler V, Bobenko A, and Suris Y} {2004} { Geometry of
  {Y}ang-{B}axter maps: pencils of conics and quadrirational mappings} {\em Comm. Anal. Geom.} {\textbf{12}} {967--1007}.
  
  \bibitem{Atkinson2006}
{\sc Atkinson J, Hietarinta J and Nijhoff F} {2007} { Seed and soliton solutions for Adler's lattice equation} {\em J. Phys. A: Math. Theor.} {\textbf{40}} {F1--F8}.

\bibitem{Berezin}
{\sc Berezin F} {1987} {Intoduction to superanalysis}  {\em Reidel Publishing Co. Dordrecht}  {\textbf{9}}.
	
\bibitem{Bobenko-Suris}
{\sc Bobenko A and Suris Y} {2002} {Integrable systems on quad-graphs}  {\em Int. Math. Res. Notices} {\textbf{11}} {573--611}.
	
\bibitem{Bridgman}
{\sc Bridgman T, Hereman W, Quispel G.R.W., van der Kamp P} {2013} {Symbolic Computation of Lax Pairs of Partial Difference Equations using Consistency Around the Cube} {\em Found. Comput. Math.}  {\textbf{13}}{517--544}.
	
	\bibitem{Caudrelier}
{\sc Caudrelier V, Cramp\'e N, Zhang C} {2014} {Integrable Boundary for Quad-Graph Systems: Three-Dimensional Boundary Consistency} {\em SIGMA} {\textbf{014} (24 pp)}.

\bibitem{Kulish}
{\sc Chaichian M  and Kulish P} {1978} {On the method of inverse scattering problem and B\"acklund transformations for supersymmetric equations} {\em Phys. Lett. B} {\textbf{78}}{413--416}.

\bibitem{Dimakis}
{\sc Dimakis A and M\"uller-Hoissen F} {2005} {An algebraic scheme associated with the non-commutative KP hierarchy and some of its extensions} {\em J. Phys. A}  {\textbf{38}}{5453--5505}.
	
	\bibitem{GKM}
{\sc Grahovski G,  Konstantinou-Rizos S and Mikhailov A} {2016} {Grassmann extensions of {Y}ang-{B}axter maps} {\em J. Phys. A} {\textbf{49}} {145202}.
	
	\bibitem{Georgi-Sasha}
{\sc Grahovski G and Mikhailov A} {2013} {Integrable discretisations for a class of nonlinear Scr\"odinger equations on {G}rassmann algebras} {\em Phys. Lett. A} {\textbf{377}} {3254--3259}.

\bibitem{Grisaru-Penati}
{\sc Grisaru M and Penati S} {2003} {An integrable noncommutative version of the sine-Gordon system} {\em Nucl. Phys. B} {\textbf{655}} {250--276}.

\bibitem{Hamanaka-Toda}
{\sc Hamanaka M and Toda} {2003} {Towards noncommutative integrable systems} {\em Phys. Lett. A} {\textbf{316}} {77--83}.

\bibitem{Hiet-Frank-Joshi}
{\sc Hietarinta J, Joshi N, Nijhoff F} {2016} {Discrete Systems and Integrability} {Cambridge texts in applied mathematics} {Cambridge University Press}.

\bibitem{Pavlos}
{\sc Kassotakis P and Nieszporski M} {2012} {On non-multiaffine consistent-around-the-cube lattice equations} {\em Phys. Lett. A} {\textbf{376}} {3135--3140}.

\bibitem{Sokor-Kouloukas}
{\sc Konstantinou-Rizos S and Kouloukas T} {2018} {A noncommutative discrete potential {K}d{V} lift} {\em J. Math. Phys.} {\textbf{59}} {063506}.

\bibitem{Sokor-Sasha-2}
{\sc Konstantinou-Rizos S and Mikhailov A} {2016} {Anticommutative extension of the Adler map} {\em J. Phys. A: Math. Theor.} {\textbf{49}} {30LT03}.

	\bibitem{Kouloukas2}
{\sc Kouloukas T and Papageorgiou V} {2011} {Poisson {Y}ang-{B}axter maps with binomial {L}ax matrices} {\em J. Math. Phys.} {\textbf{52}} {073502}.

\bibitem{Lechtenfeld}
{\sc Lechtenfeld OO, Mazzanti L, Penati S, Popov A and Tamassia L} {2005} {Integrable noncommutative sine-Gordon model} {\em Nucl. Phys. B} {\textbf{705}} {477--503}.

\bibitem{Mao-Liu}
{\sc Mao H and  Liu Q P} {2018} {B\"acklund-Darboux transformations and discretizations of $N=2$ $a=-2$ supersymmetric KdV equation} {\em Phys. Lett. A} {\textbf{381}} {253--258}.
  
 \bibitem{Mao-Liu-2}
{\sc Mao H and  Liu Q P} {2018} {Supersymmetric Sawada-Kotera equation: B\"acklund-Darboux transformations and applications} {\em J. Nonlinear Math. Phy.}
  {\textbf{25}} {375--386}.

\bibitem{NPCQ-92}
{\sc Nijhoff F W, Papageorgiou V. G., Capel H. W. and Quispel G. R. W.} {1992} {The lattice {G}el'fand-{D}ikii hierarchy} {\em Inverse Probl.} {\textbf{8(4)}} {597}.

	\bibitem{Frank4}
{\sc Nijhoff F and Walker A} {2001} {The discrete and continuous {P}ainlev{\'e} VI hierarchy and the {G}arnier systems} {\em Glasgow Math. J.}{\textbf{43}} {109--123}.

\bibitem{Frank2002}
{\sc Nijhoff F} {2002} {Lax pair for the Adler (lattice Krichever-Novikov) system} {\em Phys. Lett. A}{\textbf{297}} {49--58}.

\bibitem{Sokolov1998}
{\sc Olver P and Sokolov V} {1998} {Non-abelian integrable systems of the derivative nonlinear Schr\"odinger type} {\em Inverse Problems}{\textbf{14}} {L5--L8}.

\bibitem{Pap-Tongas}
{\sc Papageorgiou V G and Tongas A} {2009} {{Y}ang-{B}axter maps associated to elliptic curves} {\em arXiv:0906.3258v1}.

\bibitem{Pap-Tongas0}
{\sc Papageorgiou V and Tongas A} {2007} {Yang--Baxter maps and multi-field integrable lattice equations} {\em J. Phys. A: Math. Theor.} {40} {12677} 

\bibitem{Pap-Tongas-Veselov}
{\sc Papageorgiou V, Tongas A and Veselov A} {2006} {Yang--Baxter maps and symmetries of integrable equations on quad-graphs} {\em J. Math. Phys.}  {47}, 083502.

\bibitem{Veselov2}
{\sc Suris Y and Veselov A} {2003} {Lax matrices for {Y}ang-{B}axter maps} {\em J. Nonlinear Math. Phys.} {\textbf{10}} {223--230}.

\bibitem{Tongas-Nijhoff}
{\sc Tongas A and Nijhoff F} {2005} {The Boussinesq integrable system. Compatible lattice and continuum structures} {\em Glasgow Math. J.} {47} {205--219} 

\bibitem{Veselov}
{\sc Veselov A} {2003} {Yang-{B}axter maps and integrable dynamics} {\em Phys. Lett. A} {\textbf{314}} {no. 3} {214--221}.
	
\bibitem{Xue-Levi-Liu}
{\sc Xue L L, Levi D and Liu Q-P} {2013} {Supersymmetric {K}d{V} equation: {D}arboux transformation and discrete systems} {\em J. Phys. A: Math. Theor. (FTC)}  {\textbf{46}} {502001}.

\bibitem{Xue-Liu}
{\sc Xue L L and  Liu Q P} {2014} {B{\"a}cklund--{D}arboux transformations and discretizations of super {K}d{V} equation} {\em SIGMA} {\textbf{10}} paper 045.

\bibitem{Xue-Liu-2}
{\sc Xue L L and  Liu Q P} {2015} {A supersymmetric AKNS problem and its Darboux-B\"acklund transformations and discrete systems} {\em Stud. Appl. Math.}
  {\textbf{135}} {35--62}.


\end{thebibliography}
\end{document}